\documentclass[journal]{IEEEtran}

\usepackage{graphicx,epsfig,amssymb,amsmath,url,latexsym,stfloats,array,bm}
\usepackage[tight,footnotesize]{subfigure}
\usepackage{makecell}
\usepackage{mathrsfs}
\usepackage{amsthm}
\usepackage{longtable}
\usepackage{amsfonts}
\usepackage{caption}
\usepackage{graphicx}
\usepackage{multirow}
\usepackage{longtable}
\usepackage{float}
\usepackage{afterpage}
\usepackage{cite}
\usepackage{subeqnarray}
\usepackage{epstopdf}
\usepackage{empheq}
\usepackage{latexsym}
\usepackage{CJK}
\usepackage{color, soul}
\usepackage{framed}
\usepackage{lipsum}
\usepackage{color}
\usepackage{stfloats}
\usepackage{setspace}
\usepackage{hyperref}
\definecolor{shadecolor}{rgb}{1,0,0}

\newtheorem{theorem}{Theorem}

\usepackage{algorithmic}
\usepackage[linesnumbered,ruled]{algorithm2e}
\usepackage{cuted}

\begin{document}
\title{AI-Assisted Slicing-Based Resource Management for Two-Tier Radio Access Networks}

\author{ Conghao~Zhou,~\IEEEmembership{Member,~IEEE,}
            Jie~Gao,~\IEEEmembership{Senior~Member,~IEEE,}
            Mushu~Li,~\IEEEmembership{Member,~IEEE,}
            Xuemin~(Sherman)~Shen,~\IEEEmembership{Fellow,~IEEE,}
            Weihua~Zhuang,~\IEEEmembership{Fellow,~IEEE,}
            Xu~Li,~and Weisen~Shi

\thanks{C.~Zhou, X.~Shen, and W.~Zhuang are with the Department of Electrical and Computer Engineering, University of Waterloo, Waterloo, ON, N2L 3G1, Canada, (e-mail:\{c89zhou, sshen, wzhuang\}@uwaterloo.ca).}

\thanks{J.~Gao and M.~Li were with the Department of Electrical and Computer Engineering, University of Waterloo, Waterloo, ON, N2L 3G1, Canada while conducting this research work, (email:\{jie.gao, m475li\}@uwaterloo.ca).}

\thanks{X.~Li and W.~Shi are with Huawei Technologies Canada Inc., Ottawa, Ontario, K2K 3J1, Canada, (email:\{xu.lica, weisen.shi1\}@huawei.com).}
}   
\maketitle

\begin{abstract} 

While network slicing has become a prevalent approach to service differentiation, radio access network (RAN) slicing remains challenging due to the need of substantial adaptivity and flexibility to cope with the highly dynamic network environment in RANs. In this paper, we develop a slicing-based resource management framework for a two-tier RAN to support multiple services with different quality of service (QoS) requirements. The developed framework focuses on base station (BS) service coverage (SC) and interference management for multiple slices, each of which corresponds to a service. New designs are introduced in the spatial, temporal, and slice dimensions to cope with spatiotemporal variations in data traffic, balance adaptivity and overhead of resource management, and enhance flexibility in service differentiation. Based on the proposed framework, an energy efficiency maximization problem is formulated, and an artificial intelligence (AI)-assisted approach is proposed to solve the problem. Specifically, a deep unsupervised learning-assisted algorithm is proposed for searching the optimal SC of the BSs, and an optimization-based analytical solution is found for managing interference among BSs. Simulation results under different data traffic distributions demonstrate that our proposed slicing-based resource management framework, empowered by the AI-assisted approach, outperforms the benchmark frameworks and achieves a close-to-optimal performance in energy efficiency.

\end{abstract}


\section{Introduction}

Since the 3rd Generation Partnership Project (3GPP) Release~18 for the advanced fifth generation communication network (5G-advanced) in 2021, academia have commenced their efforts on the development and deployment of next-generation wireless networks (NGWNs)~\cite{chen2022standardization}. NGWNs are anticipated to support a diverse set of disruptive new services such as extended reality (XR) and haptic communications~\cite{shen2022toward}. As a result, the research and standardization efforts for NGWNs must address new challenges. First, services in NGWNs will have unprecedentedly stringent quality of service (QoS) requirements since a massive amount of data must be transmitted over networks with extremely low delay and ultra-high reliability~\cite{afolabi2018network, wang2023artificial}. Second, the QoS requirements of services in NGWNs will become highly diverse. Meeting the stringent and diverse QoS requirements to support new services in NGWNs calls for advanced networking and communication techniques~\cite{letaief2019roadmap, ma2023nomore}. 

Network slicing, as a key innovation in the fifth generation (5G), can support multiple coexisting virtual networks,~i.e.,~slices, on the same physical network infrastructure~\cite{shen2020ai}. Due to the advantages in QoS guarantee and service differentiation, network slicing lays a foundation for efficient resource management and will continue playing an important role in NGWNs. Some pioneering works have envisioned advanced slicing-based resource management for services in NGWNs with diverse and stringent QoS requirements~\cite{wang2022secure,mei2020intelligent}. In these works, slicing-based resource management can be categorized into two stages, i.e., \emph{planning stage} and \emph{operation stage}. The planning stage focuses on network-wide configuration and proactive network resource reservation for different services, while the operation stage focuses on user-level service provisioning and real-time network resource allocation~\cite{li2021slicing}. A planning period, referred to as the planning window, can be minutes or hours in length, whereas a network operation period, referred to as the operation window, is generally milliseconds in length. While planning and operation stages have different focuses, both play indispensable roles in slicing-based resource managements as they jointly determine QoS satisfaction~\cite{zhuang2021dynamic, fadlullah2022balancing}. However, existing literature and 3GPP standards pay much more attention to the operation stage than to the planning stage.

Compared to the operation stage, slicing-based resource management in the planning stage faces unique challenges. First, real-time information on individual users is unavailable at the beginning of the planning stage when resources are reserved. Consequently, existing slicing-based resource management schemes in the planning stage rely on coarse-grained information such as the aggregated data traffic over a planning window, which may result in an inaccurate estimation of service demands and thus degrade network resource utilization~\cite{zhou2022digital}. Second, user mobility and time-varying user behaviors result in significant spatiotemporal variations in service demands, which pose a challenge of balancing adaptivity and overhead in the planning stage of slicing-based resource management~\cite{shen2021holistic}. Third, differentiating services and satisfying their diverse and stringent QoS requirements further complicate the decision making on network-wide configurations and proactive resource reservation~\cite{shen2020ai}. 

Following 5G standardization in Releases~15 to~17 as well as commercial 5G deployment, a large number of works have studied slicing-based resource management for supporting diverse services in core networks~\cite{afolabi2018network}. Nevertheless, slicing-based resource management for radio access networks (RANs) is still in its infancy~\cite{mei2021intelligent,3GPP43030}. Ensuring service differentiation among multiple slices in RANs is more challenging than in core networks, and the reason is two-fold. First, interference occurs among the data transmissions of different base stations (BSs) within each slice since spectrum reuse takes place among the BSs for improving the spectrum multiplexing gain~\cite{mei2020intelligent}. Such intra-slice interference causes challenges in accurately estimating the required amount of resources for each slice, thereby adversely affecting their QoS satisfaction~\cite{zambianco2020interference}. Furthermore, inter-slice interference may occur and result in tightly coupled management (such as coverage management) among different slices in RANs, which hinders efficient slice isolation in RANs~\cite{foukas2017orion}. Therefore, slicing-based resource management for RANs that can address the aforementioned challenges needs to be further investigated in NGWNs.

In this paper, we investigate slicing-based resource management for a two-tier RAN, i.e.,~a single
macro-cell in the first tier and multiple small cells in the second tier, to improve resource utilization and achieve service differentiation. Specifically, creating a slice for each service, we determine the service coverage (SC) of BSs for each slice and manage inter-slice and intra-slice interference to support slices with different signal-to-interference-plus-noise ratio (SINR) requirements. Our research objective is to maximize the network energy efficiency by determining the SC and downlink transmission power of BSs for all slices while satisfying their SINR requirements. We propose a RAN slicing framework and formulate an optimization problem based on the proposed framework. Then, we develop an approach to solve the problem for obtaining the optimal solution of SC management (SCM) and interference management (IM). The main contributions of this paper are as follows:

\begin{itemize}

    \item We develop a novel RAN slicing framework with three designs for the spatial, temporal, and slice dimensions. The proposed grid-based planning and dual time-scale planning can adapt to spatiotemporal variations in data traffic, and the proposed flexible binary slice zooming can enhance the flexibility of service differentiation for satisfying different QoS requirements in a RAN.

    \item We propose an effective artificial intelligence (AI)-assisted approach to address the challenging RAN slicing problem. By integrating a deep unsupervised learning technique and an optimization-based analytical solution, the proposed approach can cope with the coupling between SCM and IM to balance the adaptivity and overhead of slicing-based resource management.

\end{itemize}

The remainder of this paper is organized as follows.
Section~II provides an overview of related studies.
Section~III describes the network scenario and proposed RAN slicing framework.
Section~IV presents the system model and problem formulation.
Section~V introduces the developed AI-assisted approach.
Section~VI presents the simulation results, followed by the conclusion in Section~VII. 
A list of main symbols is given in Table~I. 

\begin{table*}[t]
    \footnotesize 
    \centering
    \captionsetup{justification=centering,singlelinecheck=false}
    \caption{List of Main Symbols}\label{table1}
    \begin{tabular}{c|l|c|l}
    \hline\hline
    Symbols & \thead{Definition} & Symbols & \thead{Definition} \\
    \hline\hline
    $a_{m,n}$ & \thead[l]{The binary indicator indicating whether\\ the SC of SBS $m$ for slice $n$ is full-size\\ or reduced-size} & $p_{i,n}^{t}$ & \thead[l]{The total downlink transmission power\\ summarized over all RBs within grid $i$ for\\ slice $n$ in time interval $t$}\\
    \hline
    $b_{i,n,i',n'}^{t}$ & \thead[l]{The indicator indicating whether the\\ downlink transmission to grid $i$ for slice $n$\\ is interfered by downlink transmission to\\ grid $i'$ for slice $n'$ in time interval $t$}  & $\theta_{i,n}^{t}$ & \thead[l]{The probability that the downlink transmission\\ to grid $i$ for slice $n$ is interfered by downlink\\ transmission to grid $i'$ for slice $n'$ in time\\ interval $t$}\\
    \hline
    $E_{m,n}^{t}$ & \thead[l]{The energy consumption of slice $n$ at\\ BS $m$ in time interval $t$ } & $P_{m,n}^{t}$ & \thead[l]{The total transmission power over all grids within\\ the SC of BS $m$ for slice $n$ in time interval $t$}  \\
    \hline
    $h_{i,n}^{t}$ & \thead[l]{The average channel gain of downlink\\ transmission of BS $m_{i,n}$ over all UTs\\ within grid $i$ in time interval $t$} & $\mathcal{R}_{m}$ & \thead[l]{The set of grids within the ring-shaped area\\ surrounding SBS~$m$}\\
    \hline
    $\mathcal I_{m,n}$ & \thead[l]{The set of grids within the SC of\\ BS $m$ for slice $n$ } & $w_{i,n}^t$ & \thead[l]{The amount of downlink data traffic loads of\\ all UTs within grid $i$ in slice~$n$ in time interval $t$} \\
    \hline
    $l_{m}^\mathrm{f}$, $l_{m}^\mathrm{r}$ & \thead[l]{The full-size and reduced-size SC\\ radiuses of SBS $m$, respectively } & $\gamma_{i,n}^{t}$ & \thead[l]{The SINR of downlink transmission within\\ grid $i$ for slice $n$ in time interval $t$}\\
    \hline
    $l_{m,n}$ & \thead[l]{The SC of SBS $m$ for slice $n$} & $m_{i,n}$ & \thead[l]{The index of BS associated to grid $i$\\ for slice $n$}  \\
    \hline
    $L_\mathrm{max}$ & \thead[l]{The maximum physical coverage\\ radius of each SBS} & $\Upsilon^\text{re}$ & \thead[l]{The set of data records used for\\ solution refinement}\\
    \hline
    $r$ & \thead[l]{The diameter of each grid}  & $\Upsilon$ & \thead[l]{The set of data records}\\

    \hline

    \hline\hline

    \end{tabular}
\end{table*}

\section{Related Work}

Slicing-based resource management for core networks has attracted significant attention since 5G due to its advantage in service differentiation, while research on slicing-based resource management for RANs is still at a nascent stage~\cite{afolabi2018network}. Existing works on slicing-based resource management for RANs can be categorized as either a \emph{single-stage} approach or a \emph{two-stage} approach (i.e., with planning and operation stages as mentioned in Section I).

In single-stage approaches, a centralized controller, e.g., a software-defined networking (SDN) controller, is responsible for managing resources in a RAN for each individual user terminal (UT) in each slice~\cite{korrai2020ran,caballero2017multi,yang2020mixed,sun2020service,oladejo2020latency,xiang2020realization,zambianco2020interference}. In a single-BS scenario, Korrai~\emph{et. al} focused on the physical-layer RAN slicing and investigated customized physical-layer configurations for UTs in different slices~\cite{korrai2020ran}, while Yang~\emph{et. al} concentrated on the data link layer and proposed a resource block (RB) scheduling scheme for UTs of enhanced mobile broadband (eMBB) and ultra-reliable and low latency communications (URLLC) slices to satisfy their different latency and reliability requirements~\cite{yang2020mixed}. In a multiple-BS scenario, authors in~\cite{caballero2017multi} proposed an orthogonal RB allocation scheme for UTs in different slices from the perspective of fairness in data rates of UTs. Moreover, with the consideration of inter-slice and intra-slice interference, a few works presented RB allocation schemes for UTs in different slices to improve their performance in terms of latency, data rate, and RB usage~\cite{sun2020service,oladejo2020latency,xiang2020realization,zambianco2020interference}. While single-stage approaches can support service differentiation, their adaptivity is restricted owing to the lack of proactive resource reservation, which poses a challenge to QoS guarantee in highly dynamic network environments~\cite{filali2022dynamic}.

To tackle this problem, lots of researchers recently concentrate on two-stage approaches~\cite{shen2020ai,mei2020intelligent}. Specifically, a centralized controller proactively reserves network resources for slices according to the service demand of each slice in a large time scale, i.e., planning window, whereas each slice allocates the reserved resources to individual UTs based on their real-time status in a short time scale, i.e., operation window. Compared with one-stage approaches, two-stage approaches are capable of achieving high adaptivity by proactively configuring slices and reserving resources in dynamic network environments and offer great flexibility due to having two time scales for different resource management decisions~\cite{shen2020ai,adamuz2022stochastic}. Focusing on the planning stage, a few existing works investigated proactive resource reservation in RAN slicing by statistically modeling the service demand of each slice~\cite{wu2020dynamic,jankovic2021effects,adamuz2022stochastic,d2019slice}, e.g., Poisson process-based data packet arrival. Considering vehicular networks with eMBB, URLLC, and massive machine-type communication (mMTC) services, the authors of~\cite{wu2020dynamic} and~\cite{jankovic2021effects} proposed two orthogonal radio resource reservation schemes, respectively. Taking into account inter-slice interference, some researchers presented spectrum slicing schemes, e.g.,~\cite{adamuz2022stochastic,d2019slice}, and the authors of~\cite{pablo2020radio} analyzed the trade-off between spectrum utilization and inter-slice interference. In addition, joint planning-stage and operation-stage radio resource slicing was studied in various network scenarios, including one-tier~\cite{bakri2021data,yan2019intelligent}, two-tier~\cite{ye2018dynamic}, and drone-based RANs~\cite{shen2021drone}, where radio resource reservation among slices in the planning stage was conducted based on AI-driven prediction~\cite{bakri2021data,yan2019intelligent} or statistical modeling~\cite{ye2018dynamic,shen2021drone} of the data traffic load in each slice. Existing research on two-stage approaches mainly concentrated on resource reservation for RANs, while SC management for multiple slices with different QoS requirements remains an open issue. Moreover, the existing two-stage approaches rely on coarse-grained information, such as aggregated data traffic within the SC area of a BS, which may degrade network resource utilization.

Different from the existing two-stage approaches, we propose a novel RAN slicing framework for both resource reservation and SC management in the planning stage. With joint resource reservation and SC management, we target fine-grained and flexible resource management for achieving service differentiation in spatiotemporally dynamic network environments.

\section{Network Scenario and RAN Slicing Framework}

In this section, we introduce the considered network scenario and present the proposed RAN slicing framework.

\subsection{Network Scenario}

Consider a two-tier RAN with one macro BS (MBS) in the first tier and~$M$ small BSs (SBSs) in the second tier. We show the physical network scenario of the considered two-tier RAN in Fig.~\ref{system}. All the BSs use the same radio spectrum pool, and each BS orthogonally reserves RBs for downlink transmissions within its coverage area~\cite{3GPP36872}. Using network slicing, $N$ slices (corresponding to~$N$ services with different SINR requirements) are created on top of the physical network, and the radio spectrum resource of each BS is shared by all slices. For each slice, the MBS and all SBSs jointly support the corresponding service across the network to ensure that the service is accessible anywhere within the network coverage area. Meanwhile, given any slice, the SC of different SBSs (representing the spatial coverage of these SBSs for the corresponding service) are non-overlapping with each other for mitigating intra-slice interference. Any UT within the SC of an BS is associated to that BS, and each BS solely serves all UTs within its SC for the corresponding service. UTs not within the SC of any SBS are associated to the MBS. A centralized controller located at the MBS determines the SC and total transmission power for each slice at each BS in the planning stage, corresponding to SCM and IM. Then, in the subsequent operation stage, each BS allocates radio resources, such as RBs and transmission power, to individual UTs within its SC for downlink transmissions.

    \begin{figure}[t]
        \centering
        \includegraphics[width=0.40\textwidth]{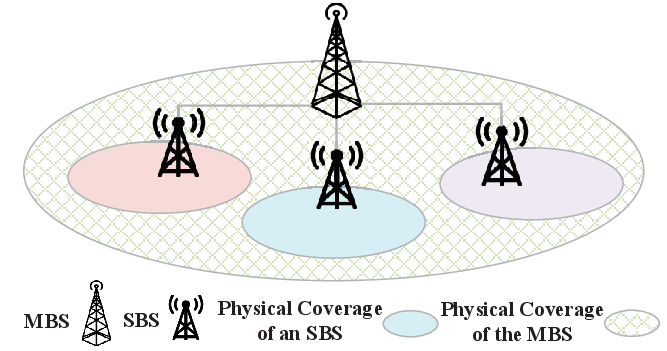}
        \caption{The physical network scenario.}\label{system}
    \end{figure}

\subsection{RAN Slicing Framework}

For the considered scenario, we focus on the planning stage and propose a RAN slicing framework to achieve fine-grained and flexible SCM and IM for services with different SINR requirements. The proposed framework consists of three schemes: 1) grid-based planning in the spatial dimension; 2) dual-time scale planning in the temporal dimension; and 3) flexible binary slice zooming in the slice dimension. 

\subsubsection{Grid-based Planning} 

    \begin{figure}[t]
      \centering
        \subfigure[An illustration of grids and SC.]
        {\includegraphics[width=0.45\textwidth]{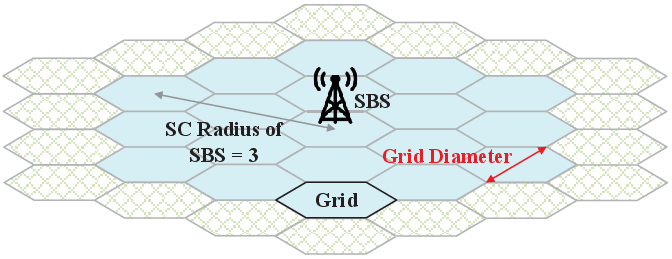}}
        \subfigure[An illustration of the time scales.]
        {\includegraphics[width=0.49\textwidth]{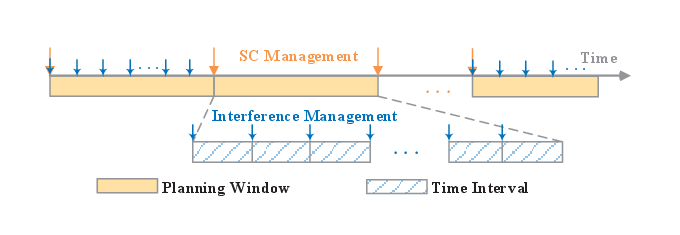}} \\
      \caption{Grid-based and dual time-scale planning (shown for one slice).}
        \label{f:grid}
    \end{figure} 

To cope with the uneven spatial distribution of data traffic loads, we propose grid-based planning in the spatial dimension, where the illustration of grid-based planning for one slice is shown in Fig.~\ref{f:grid}(a). Specifically, the whole network coverage is divided into $I$ hexagon areas, named \emph{grids}, with an identical grid diameter, denoted by~$r$.\footnote{In addition to hexagons, some other shapes of grids are also applicable to the proposed grid-based planning.} We assume that each BS is at the center of a grid, and the SC radius of each SBS corresponds to the number of layers of grids within its SC. For each slice, the SC of the MBS includes all the grids that are not in the SC of any SBS. In the example shown in Fig.~\ref{f:grid}(a), the SC radius of the SBS is~$3$. The total downlink transmission power for each grid within the SC of a BS can be different.

The benefit of grid-based planning is two-fold. First, the downlink transmissions for UTs within different grids may experience different interference. Customizing the total transmission power for each grid can help mitigate inter-slice and intra-slice interference and thus improve network energy efficiency. Second, adjusting the SC of each SBS in the units of grids is beneficial for balancing data traffic loads among BSs in a fine-grained manner.

\subsubsection{Dual Time-scale Planning}

To adapt to the temporal variations of data traffic loads, we propose the scheme of dual time-scale planning, where the illustration of dual time-scale planning for one slice is shown in Fig.~\ref{f:grid}(b). Each planning window is divided into $T$ ($T>1$) \emph{time intervals} with uniform length. The SC of the SBSs is updated at the beginning of each planning window and remains constant till the beginning of the next planning window. By contrast, the total downlink transmission power of the BSs for individual grids is updated at the beginning of each time interval in the planning window if needed. 

Dual time-scale planning provides great flexibility in differentiating the time scales of SCM and IM based on the difference in the amount of resource management overhead, i.e., signaling overhead and computation complexity. First, the short planning window for SCM leads to frequent UT association changing during network operations and thus high signaling overhead. Second, SCM has a higher computation complexity than IM since adjusting the SC of BSs results in the update of total transmission power of BSs for individual grids. Dual-time scale planning helps properly balance the adaptivity of resource management and the resource management overhead.

\subsubsection{Flexible Binary Slice Zooming}

\begin{figure}[t]
    \centering
    \includegraphics[width=0.48\textwidth]{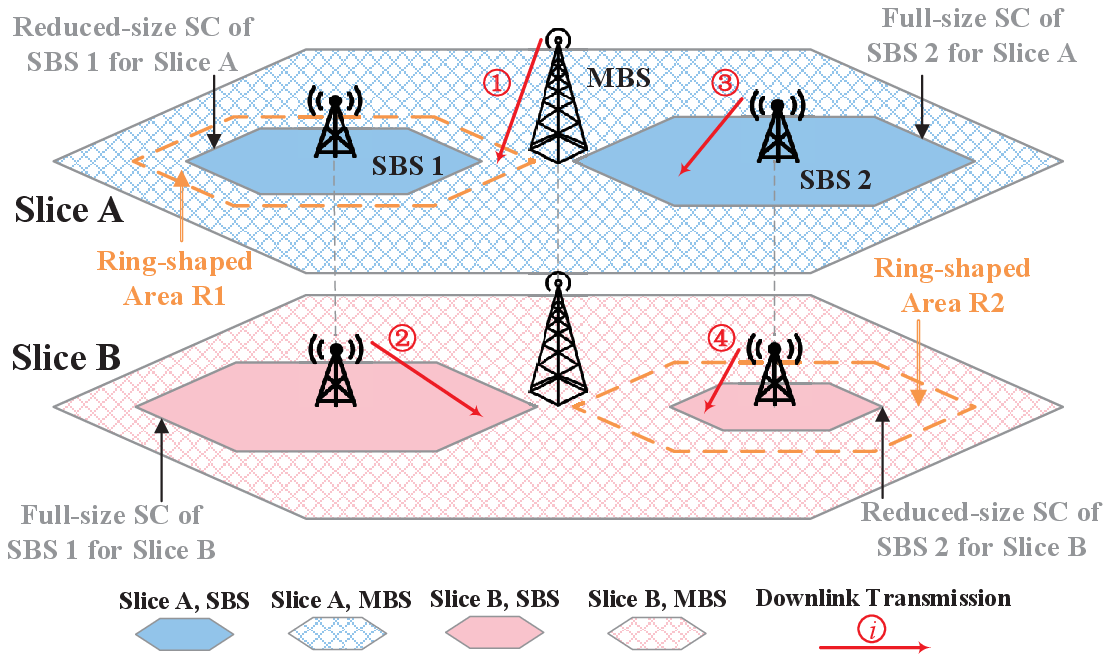}
    \caption{Flexible binary slice zooming (shown for two slices).}\label{f:FBSL}
\end{figure}

To provide the flexibility in service differentiation, we propose a novel scheme called \emph{flexible binary slice zooming} in the slice dimension, including the following two elements. The proposed scheme for two slices is illustrated in Fig.~\ref{f:FBSL}.
    \begin{itemize}
        \item \textbf{Differentiated IM and SCM across slices:} For IM, the transmission power reserved by each BS for each grid can be different across slices. For SCM, the SC of each SBS can also be different across slices. Specifically, the SC of each SBS for any slice is binary, i.e., either full-size or reduced-size shown in Fig.~\ref{f:FBSL}, neither of which can exceed the maximum physical coverage area of the SBS. We refer to the gap between the full-size and the reduced-size SC of each SBS as a \emph{ring-shaped area} surrounding the SBS.\footnote{All SC of an SBS may be identical, i.e., either all SC is reduced-size or all SC is full-size. In this case, there is no ring-shaped area surrounding the SBS.}

        \item \textbf{Partially non-orthogonal RB reservation among BSs:} All the BSs use the same radio spectrum pool except in the following case: each SBS and the MBS reserve different sets of RBs for downlink transmissions within the ring-shaped area surrounding the SBS if such an area exists. The partially non-orthogonal RB reservation avoids the interference between the downlink transmissions of each SBS and the MBS within the ring-shaped areas. We highlight the partially non-orthogonal RB reservation for downlink transmissions using red arrows in Fig.~\ref{f:FBSL}, which is explained in Subsection~IV.A. 

    \end{itemize} 

The proposed flexible binary slice zooming has two benefits in facilitating service differentiation in RANs. First, differentiating the downlink transmission power for different slices achieves fine-grained IM in the slice dimension, and thus helps satisfy the diverse and stringent SINR requirements of slices. Second, by customizing the SC of each BS for different slices, the proposed flexible binary slice zooming scheme is more flexible in adapting to the different spatial distributions of data traffic loads than conventional cell-based SCM that uses identical SC for all services.

Based on the aforementioned three schemes, the proposed RAN slicing framework provides great flexibility in enabling isolated SCM and IM for multiple slices, and improves granularity and adaptivity in adapting to the spatiotemporal variations of data traffic loads in RANs.

\subsection{Operation Stage Consideration}

    \begin{figure}[t]
        \centering
        \subfigure[Virtual slice separation in the planning stage.]
        {\includegraphics[width=0.24\textwidth]{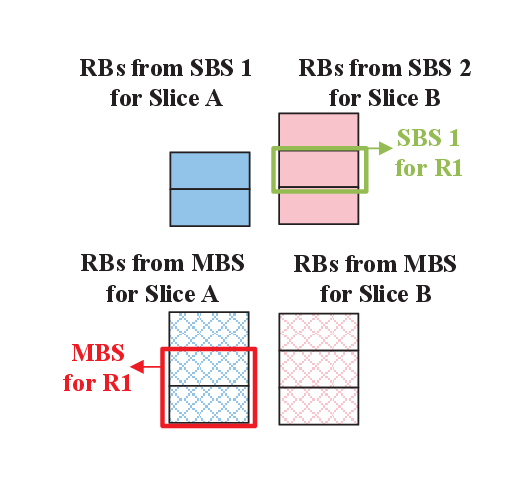}} 
        \subfigure[RB allocation in the operation stage.]
        {\includegraphics[width=0.24\textwidth]{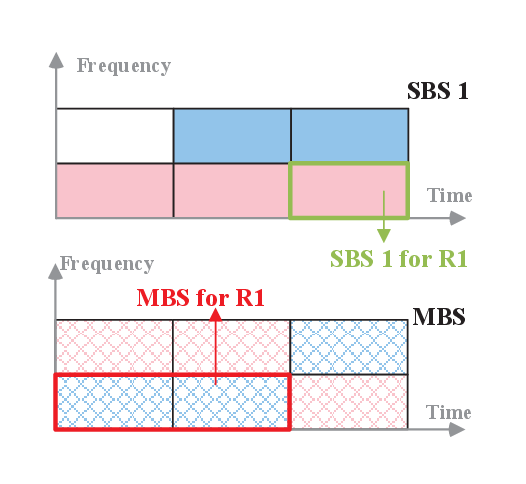}}
        \caption{Virtual slice separation at SBS~1.}
        \label{f:zooming-virtual}
    \end{figure}

The real-time allocation of RBs for individual UTs in the operation stage impacts the interference and thus the SINR of each UT. As a result, making decisions on SCM and IM with the consideration of operation-stage RB allocation is necessary. However, it is impossible to know the future UT-level information, e.g., the locations and data traffic loads of UTs, at the beginning of each planning window, and how RBs will be allocated to individual UTs in the subsequent planning window. To overcome this issue, we adopt \emph{virtual slice separation} to reserve RBs for multiple slices in the planning stage with RB multiplexing in the operation stage. Specifically, only the number of RBs reserved for each slice is determined in the planning stage rather than the specific set of RBs. An example of virtual slice separation is shown in Fig.~\ref{f:zooming-virtual}. Fig.~\ref{f:zooming-virtual}(a) shows the numbers of RBs reserved to slices~A and~B by virtual slice separation, and Fig.~\ref{f:zooming-virtual}(b) shows the specific sets of RBs that can be flexibly allocated to individual UTs in the operation stage based on the real-time network environment.

\section{System Model and Problem Formulation} 

In this section, we first present the system model for SCM and IM based on the proposed RAN slicing framework. Then, we formulate an optimization problem to maximize the network energy efficiency. 

Denote the set of BSs by $\mathcal M = \left\{0,1,\cdots, M \right\}$, and let~$m = 0$ and~$m \in \mathcal{M}\backslash\{0\}$ be the indexes of the MBS and~$M$ SBSs, respectively. Define the sets of slices, grids, and time intervals as $\mathcal N = \left\{1,2,\cdots, N \right\}$, $\mathcal{I} = \left\{1,2, \cdots, I \right\}$, and $\mathcal T$ = $\{1, 2, \cdots, T \}$, respectively.

\subsection{Model of SCM}

We model the SC of BSs in the proposed RAN slicing framework. Denote the SC radius of SBS~$m$ for slice~$n$ by $l_{m,n}$. We assume that the maximum physical coverage radius of all SBSs are identical, denoted by $L_\mathrm{max}$, and define the set of possible SC radius for any slice as $\mathcal{L} = \{1,2, \cdots, L_\mathrm{max}\}$. With flexible binary slice zooming, we determine the full-size or reduced-size SC radiuses of SBS~$m$, denoted by $l_{m}^\mathrm{f} \in \mathcal{L}$ and $l_{m}^\mathrm{r} \in \mathcal{L}$, respectively, where  $l_{m}^\mathrm{f} \geq l_{m}^\mathrm{r}$. To indicate whether the SC of SBS $m$ for slice $n$ is full-size or reduced-size, we introduce a binary variable~$a_{m,n} \in \{0, 1\}$.  Accordingly, $\mathbf{l}^\mathrm{f}  = [l_{m}^\mathrm{f}]_{\forall m \in \mathcal M \backslash \{ 0 \}}$, $\mathbf{l}^\mathrm{r}  = [l_{m}^\mathrm{r}]_{\forall m \in \mathcal M \backslash \{ 0 \}}$, and $\mathbf{a}  = [a_{m,n}]_{\forall m \in \mathcal{M} \backslash \{ 0 \}, n \in \mathcal{N}}$ are the variables that determine the SC of BSs during a planning window. The SC radius of SBS~$m \in \mathcal{M}\backslash\{0\}$ for slice~$n \in \mathcal{N}$ is represented as follows:
	\begin{equation}\label{eq1}
		l_{m,n} =\left\{
		\begin{aligned}
		&l_{m}^\mathrm{f}, \,\, \text{if} \,\, a_{m,n} = 1;\\
		&l_{m}^\mathrm{r}, \,\, \text{if} \,\, a_{m,n} = 0.
		\end{aligned}
		\right.
	\end{equation}
 Let $d_{m,i}$ denote the distance between BS~$m$ and the center of grid~$i$. We define the set of grids within the SC of SBS~$m$ for slice~$n$ and the set of grids within the ring-shaped area surrounding SBS~$m$ as~$\mathcal{I}_{m,n} = \{i| d_{m,i} \leq l_{m,n}, i \in \mathcal{I} \}$ and~$\mathcal{R}_{m} = \{i | l_{m}^\mathrm{r} \leq d_{m,i} \leq l_{m}^\mathrm{f}, i \in \mathcal{I} \}$, respectively. We define the set of grids within the SC of the MBS for slice~$n$ as $\mathcal{I}_{0,n} = \{i | i \in \mathcal{I} \backslash \mathcal{I}_{m,n}, m \in \mathcal{M} \backslash \{0\} \}$, i.e., for any slice, grids that are not within the SC of any SBS are covered by the MBS.

SCM should consider the spatial distribution of downlink data traffic loads. Denote the amount of downlink data traffic loads (in bits) of all UTs within grid~$i$ in time interval~$t$ for slice~$n$ by~$w_{i,n}^{t}$. Let vectors~$\mathbf{w}_{n}^{t} = [w_{i,n}^{t}]_{\forall i \in \mathcal{I}}$ and~$\mathbf{W} = [w_{i,n}^{t} ]_{\forall i \in \mathcal{I}, n \in \mathcal{N}, t \in \mathcal{T}}$ be the data traffic distribution (DTD) of slice~$n$ in time interval~$t$ and the DTD vector of all slices over a planning window, respectively. The DTD vector~$\mathbf{W}$ is assumed to be known~\emph{a prior} through prediction~\cite{osseiran20165g}. The required number of RBs for each BS depends on the data traffic load within the SC of the BS. Let $\eta_{n}$ represent the average number of RBs required to support each bit of the downlink data traffic in slice~$n$\footnote{The value of $\eta_{n}$ can be estimated according to the data rate requirement of slice~$n$ and the long-term performance of RB scheduling in the operation stage~\cite{yang2008proactive}.} To ensure that the number of RBs reserved for the downlink data traffic within the SC of any BS does not exceed the total number of RBs of each BS during a planning window, denoted by $C$, the condition
    \begin{equation}\label{eq2p}
        \sum_{t \in \mathcal{T}}{\sum_{n \in \mathcal{N}}{ \sum_{i \in \mathcal{I}_{m,n}}{ \eta_{n}  w_{i,n}^{t} } }}  \leq C, \,\, \forall m \in \mathcal{M}\backslash\{0\},
    \end{equation}
for each SBS and the condition
    \begin{equation}\label{eq3p}
       \sum_{t \in \mathcal{T}}{\sum_{m \in \mathcal{M}\backslash\{0\}}{\sum_{n\in\mathcal{N}}{ \eta_{n} (\sum_{i \in \mathcal{I}}{w_{i,n}^{t} } - \sum_{i \in \mathcal{I}_{m,n}}{   w_{i,n}^{t}  } )} } } \leq C,
    \end{equation}
for the MBS should be satisfied in SCM.

SCM affects the interference among downlink transmissions of different BSs due to frequency reuse. As the result of partially non-orthogonal RB reservation in flexible binary slice zooming, there are two cases in which downlink transmissions do not interfere with each other: i) between downlink transmissions of the same BS, e.g., communication links~3 and~4 in Fig.~\ref{f:FBSL} (shown as the red arrows with circled numbers~3 and~4 in the figure); and ii) between downlink transmissions of the SBSs and the MBS within a ring-shaped area, e.g., communication links~1 and~2 for the two UTs in the ring-shaped area~R1 in Fig.~\ref{f:FBSL} (shown as the red arrows with circled numbers~1 and~2). To achieve non-orthogonal RB reservation as mentioned in Subsection~III.B, the following condition must be satisfied in SCM: 

    \begin{equation}\label{eq2}
          \sum_{t \in \mathcal{T}}{\sum_{n \in \mathcal N}{ \sum_{i \in \mathcal{R}_{m}}{ w_{i,n}^{t} \eta_n } }} \leq C, \,\,  \forall m \in \mathcal{M}\backslash\{0\}.
    \end{equation}

Constraint~\eqref{eq2} ensures that, if a ring-shaped area surrounding an SBS exists, the SBS and the MBS have sufficient RBs for orthogonal RB reservation in the ring-shaped area. 

Other than the aforementioned two cases, there exists the interference between downlink transmissions of different BSs. We introduce term $b_{i,n,i',n'}^{t} \in \{0, 1 \}$ to indicate whether the downlink transmission to grid~$i$ for slice~$n$ is interfered by the downlink transmission to grid~$i'$ for slice~$n'$ in time interval~$t$ or not, given by:
    \begin{equation}\label{eq3}
        \begin{aligned}
            & b_{i,n,i',n'}^{t} = \\
            & \left\{ \begin{array}{l}             
                0,\,\, \text{if}\,\,  i \in \mathcal{I}_{m,n}, i' \in \mathcal{R}_{m}, a_{m,n} =1,  a_{m,n'}=0;\\
                0,\,\, \text{if}\,\,  i \in \mathcal{R}_{m}, i' \in \mathcal{I}_{m,n'}, a_{m,n} =0,  a_{m,n'}=1;\\
                0,\,\, \text{if}\,\, m_{i,n} = m_{i', n'}, \forall  m_{i,n}, m_{i', n'} \in \mathcal{M};\\
                1,\,\, \text{otherwise},
            \end{array} \right.
        \end{aligned}
    \end{equation}
where $m \in \mathcal{M}\backslash\{0\}$, and $m_{i,n} = \{ m | i \in \mathcal{I}_{m,n}, m \in \mathcal{M} \}$ denotes the BS that covers grid~$i$ in its SC of slice~$n$. The first and second cases in~\eqref{eq3} represent no interference between the downlink transmission of the SBSs and the MBS within the ring-shaped areas. The third case in~\eqref{eq3} represents no interference between the downlink transmissions within grid~$i$ for slice~$n$ and that within gird~$i'$ for slice~$n'$ if they are from the same BS, i.e., $m_{i,n} = m_{i', n'}$. Otherwise, the downlink transmission within a grid interferes with the downlink transmission within other grids.     

\subsection{Model of IM}

We model IM based on virtual slice separation mentioned in Subsection~III.C. Denote the total downlink transmission power summarized over all RBs reserved for downlink transmissions in grid~$i$ for slice~$n$ in time interval~$t$ by~$p_{i,n}^{t}$. Define the IM decision in time interval~$t$ and in a planning window as $\mathbf{p}^t  = [p_{i,n}^{t}]_{\forall i \in \mathcal{I}, n \in \mathcal{N}}$ and $\mathbf{p} = [p_{i,n}^{t}]_{\forall i \in \mathcal{I}, n \in \mathcal{N}, t \in \mathcal{T}}$, respectively. 

We assume that the maximum downlink transmission power of SBSs are the same, and denote the maximum downlink transmission power of the SBSs and the MBS by~$p_\mathrm{SBS}$ and~$p_\mathrm{MBS}$, respectively. The following constraint should be satisfied in IM to ensure that the total downlink transmission power of each BS over all slices cannot exceed the maximum downlink transmission power of the BS:
    \begin{equation}\label{eq6}
        \sum_{n\in \mathcal{N}}{\sum_{i \in \mathcal I_{m,n}}{ p_{i,n}^{t} }} \leq \left\{
        \begin{aligned}
        &p_\mathrm{MBS}, &&\;  \,\,\, m = 0;\\
        &p_\mathrm{SBS}, &&\;  \,\,\, m \in \mathcal{M} \backslash \{0\}.
        \end{aligned}
        \right.
    \end{equation} 

Next, we model the interference between downlink transmissions of different BSs. The exact interference depends on real-time RB scheduling during the operation stage and is unknown~\emph{a priori} in the planning stage. We define parameter~$\theta_{i,n,i',n'}^{t} \in [0,1]$ to represent the likeliness that the downlink transmission to grid~$i$ for slice~$n$ is interfered by the downlink transmission to grid~$i'$ for slice~$n'$ in time interval~$t$ and model the planning-stage interference statistically~\cite{3GPP363,pablo2020radio}.\footnote{The value of parameter~$\theta_{i,n,i',n'}^{t}$ can be obtained empirically when the DTD~$\mathbf{W}$, RB scheduling policy in the operation stage, and RB reservation policy in the planning stage are given~\cite{adamuz2022stochastic}.} The SCM decisions of the BSs covering grid~$i$ and $i'$ for slice~$n$ and $n'$ can affect the data traffic loads of the BSs and thus the value of~$\theta_{i,n,i',n'}^{t}$. Given~$\theta_{i,n,i',n'}^{t}$, the total interference to the downlink transmission of BS~$m_{i,n}$ to grid~$i$, denoted by~$I_{i,n}^{t}$, is expressed as follows:
    \begin{equation}\label{}
        I_{i,n}^{t} = \sum_{n' \in \mathcal N}{ \sum_{i' \in \mathcal{I} }{b_{i,n,i',n'}^{t} \theta_{i,n,i',n'}^{t} p_{i',n'}^{t}  h_{i, n,i',n'}^{t}} }.
    \end{equation}
where $h_{i,n,i',n'}^{t}$ denotes the average channel gain of the downlink transmission of BS~$m_{i',n'}$ to grid~$i$ for slice~$n$ in time interval~$t$. 

The SINR of the downlink transmission of BS~$m_{i,n}$ to grid~$i$ for slice~$n$ in time interval~$t$, denoted by $\gamma_{i,n}^{t}$, can be modeled as follows:
    \begin{equation}\label{eq7}
        \gamma_{i,n}^{t} = \frac{ \bar{p}_{i,n}^{t} h_{i,n,i,n}^{t} }{N_{0} + I_{i,n}^{t} },\,\,\,\, \forall i \in \mathcal I, n \in \mathcal N, t \in \mathcal{T},
    \end{equation}
where $\bar{p}_{i,n}^{t} = p_{i,n}^{t} / (w_{i,n}^{t} \eta_{n} )$ represents the average transmission power on a single RB for downlink transmissions in slice~$n$ within grid~$i$ in time interval~$t$, and $N_{0}$ denotes the noise power. IM should satisfy the SINR requirement of each slice, as follows:
    \begin{equation}\label{eq9}
        \gamma_{i,n}^{t} \geq \rho  \gamma_n^\mathrm{min}, \;\;\;\;  \forall i \in \mathcal I_{m,n},
    \end{equation}
where $\gamma_n^\mathrm{min}$ denotes the minimum SINR required by slice~$n$, and~$\rho$ is a constant used for flexibly scaling the minimum required SINR level~\cite{3GPP43030}.\footnote{The SINR in the planning stage, i.e., $\gamma_{i,n}^{t}$ is a reference value over the duration of a time interval, which may not represent the exact SINR level in the operation stage. Thus, we allow a feedback mechanism to change the SINR requirements of slices in the planning stage by adjusting weight $\rho$ based on the real-time power control, RB allocation, and instantaneous SINR in the operation stage.}

\subsection{Problem Formulation}

In this subsection, we formulate an energy efficiency maximization problem based on the proposed RAN slicing framework. Denote the energy consumption of BS~$m$ for serving slice~$n$ in time interval~$t$ by~$E_{m,n}^{t}$, given by:
    \begin{equation}\label{eq8}
        E_{m,n}^{t} = \tau P_{m,n}^{t}, \;\;\;\; \forall m \in \mathcal M, n \in \mathcal N, t \in \mathcal T,
    \end{equation}
where $\tau$ denotes the duration of each time interval. The energy efficiency (measured in the unit of bit/RB/J) of all BSs for serving slice~$n$ during a planning window, denoted by~$\xi_{n}$, is as follows: 
    \begin{equation}
        \xi_{n}  = \frac{w_{n}}{E_{n} C_{n}}, \,\, \forall n \in \mathcal{N},
    \end{equation}
where $E_{n} = \sum_{ t \in \mathcal T}{\sum_{ m \in \mathcal M}{E_{m,n}^{t} }}$ is the total energy consumption of all BSs in all time intervals of a planning window, $w_{n} = \sum_{ t \in \mathcal T}{\sum_{ m \in \mathcal M}{ \sum_{ i \in \mathcal{I}_{m,n}}{w_{i,n}^{t}} }}$ represents the total downlink traffic data loads in the planning window, and $C_{n} = \sum_{ t \in \mathcal T}{\sum_{ m \in \mathcal M}{\sum_{ i \in \mathcal{I}_{m,n}}{w_{i,n}^{t} \eta_{n} }}}$ is the total number of RBs reserved in the planning window.

The slicing-based resource management problem with the objective of network energy efficiency maximization is formulated as follows:
    \begin{subequations}\label{p1}
        \begin{align}
            \textrm{P1:} \,\, & \max_{ \{ \mathbf{p}, \mathbf{l}^\mathrm{f}, \mathbf{l}^\mathrm{r}, \mathbf{a} \} } \sum_{ n \in \mathcal N}{ \lambda_{n} \xi_{n} }\\
            \textrm{s.t.} & \,\, ,\eqref{eq2p}, \eqref{eq3p}, \eqref{eq2}, \eqref{eq6}, \eqref{eq9}, \\
            & \,\, D_{m, m'} \geq l_m^\mathrm{f} + l_{m'}^\mathrm{f}, \;\; \forall  m \neq m', \,\, m, m' \in \mathcal M \backslash \{ 0 \},\\          
            & \,\, p_{i,n}^{t} > 0 , \;\; \forall p_{i,n}^{t} \in \mathbb R,\\
            & \,\,l_m^\mathrm{f} \geq l_m^\mathrm{r}, \;\; \forall l_m^\mathrm{r}, l_m^\mathrm{f} \in \mathcal{L}, m \in \mathcal{M} \backslash\{ 0 \},\\
            & \,\, a_{m,n} \in \left\{0,1 \right\}, \;\;  \forall m \in \mathcal M,  n \in \mathcal N,
        \end{align}
    \end{subequations}
where $\lambda_{n}$ denotes the weight for balancing the energy efficiency for different slices. In Problem~P1, the optimization variables include IM decision $\mathbf{p}$ and SCM decisions $\mathbf{l}^\mathrm{f}$, $\mathbf{l}^\mathrm{r}$, and $\mathbf{a}$. Constraint~(\ref{p1}c) ensures that the SC of SBSs does not overlap, in which term~$D_{m, m'}$ denotes the physical distance between SBSs~$m$ and~$m'$. Constraint~(\ref{p1}d) guarantees that the downlink transmission power is positive. Constraints~(\ref{p1}e) and (\ref{p1}f) ensure that the selection of the SC of each SBS for each slice is binary and does not exceed the maximum physical coverage of the SBS. Problem~P1 is a combinatorial optimization problem, which is difficult to solve by conventional optimization methods due to two reasons~\cite{korte2011combinatorial}. First, a large number of variables need to be determined. Specifically, the variables for transmission power and SCM are with the dimensions of $N \times I \times T$ and $N \times M$, respectively. Second, the transmission power and SCM decisions are coupled. To solve this problem, we propose an unsupervised learning-assisted solution in the next section.

\section{Unsupervised-learning-assisted Solution}

We decouple Problem~P1 into two sub-problems and solve them in two steps. In the first step, we design an unsupervised learning-assisted approach to determine the SC of the SBSs. In the second step, given a solution to the SCM sub-problem, we derive the closed-form solution to the IM sub-problem in each time interval. We first discuss the solution to IM in Subsection~V.A, followed by the solution to SCM in Subsections~V.B and V.C.

\subsection{Optimal Solution of IM} 

Given the settings of the SC of all SBSs, i.e., $\mathbf{l}_\mathrm{f}, \mathbf{l}_\mathrm{r}, \mathbf{a}$, we formulate the problem of IM in time interval~$t$ as follows:
    \begin{subequations}\label{p2}
        \begin{align}
            \textrm{P2:} \,\, & \max_{ \{ \mathbf{p}^t \} }  \sum_{ n \in \mathcal N}{ \lambda_{n} \xi_{n} }\\
            \textrm{s.t.}& \,\, (\ref{eq6}), (\ref{eq9}), (\ref{p1}\text{d}).
        \end{align}
    \end{subequations}  
The solution of $\mathbf{p}^t$ in Problem~P2 depends on the DTDs of all slices in time interval~$t$. In Theorem~\ref{theorem1}, we provide the closed-form optimal solution of~$\mathbf{p}^t$ in time interval~$t$. Theorem~\ref{theorem1} can be applied to all time intervals of a planning window since IM in different time intervals is independent.  

    \begin{figure*}[b]
        \rule[-10pt]{18.15cm}{0.05em}   
        \begin{equation}\label{eq12}
            \begin{split}
                & \mathbf{p}^t_{*} = \rho \left(    \hat{\mathbf{H}}^t - \rho \left[ \begin{matrix}
                    \gamma _{1}^{\min} \mathbf{\Omega }_{1,1}^{t} &        ...&        \gamma _{1}^{\min}  \mathbf{\Omega }_{1,n'}^{t} &        ...&        \gamma _{1}^{\min}\mathbf{\Omega }_{1,N}^{t} \\
                    \vdots&        \ddots&        ...&        \ddots&         \vdots\\
                    \gamma _{n}^{\min}\mathbf{\Omega }_{n,1}^{t} &        ...&        \gamma _{n}^{\min}\mathbf{\Omega }_{n,n'}^{t} &        ...&        \gamma _{n}^{\min}\mathbf{\Omega }_{n,N}^{t} \\
                    \vdots&        \ddots&        ...&        \ddots&        \vdots\\
                    \gamma _{N}^{\min}\mathbf{\Omega }_{N,1}^{t} &        ...&        \gamma _{N}^{\min}\mathbf{\Omega }_{N,n'}^{t} &        ...&        \gamma _{N}^{\min}\mathbf{\Omega}_{N,N}^{t} \\
                \end{matrix} \right]_{IN \times IN} \right)^{-1} \left[ \begin{matrix}
                    \gamma_{1}^{\min} N_0\\
                    \vdots\\
                    \gamma_{n}^{\min} N_0\\
                    \vdots\\
                    \gamma_{N}^{\min} N_0
                    \end{matrix} \right]_{IN \times 1},
            \end{split}
        \end{equation}
    \end{figure*}

    \begin{theorem}\label{theorem1}
        Define $\delta_{i,n,i',n'}^{t} = b_{i,n,i',n'}^{t} \theta_{i,n,i',n'}^{t} h_{i, n,i',n'}^{t}$. The optimal solution to Problem~P2, i.e., $\mathbf{p}^t_{*}$, is given by~\eqref{eq12}, where
        	\begin{equation}\label{eq13}
                \begin{aligned}
                    & \mathbf{\Omega}_{n,n'}^{t} = \left[\begin{matrix}
                            \delta_{1,n,1,n'}^{t} &      \cdot&     \delta_{1,n,i',n'}^{t} &        \cdot&     \delta_{1,n,I,n'}^{t} \\
                            \vdots&     &     \vdots&     &       \vdots\\
                            \delta_{i,n,1,n'}^{t} &      \cdot&     \delta_{i,n,i',n'}^{t} &        \cdot&     \delta_{i,n,I,n'}^{t} \\
                            \vdots&     &       \vdots&     &     \vdots\\
                            \delta_{I,n,1,n'}^{t} &      \cdot&     \delta_{I,n,i',n'}^{t} &        \cdot&     \delta_{I,n,I,n'}^{t} \\
                        \end{matrix}\right]_{I \times I},
                \end{aligned}
            \end{equation}
        and 
    \begin{equation}\label{eq18p}
        \hat{\mathbf{H}}^t = \mathrm{diag}\left(\frac{h_{1,1}^{t}}{w_{1,1}^{t} \eta_{1}}, \cdots, \frac{h_{i,n}^{t}}{w_{i,n}^{t} \eta_{n}}, \cdots, \frac{h_{I,N}^{t}}{{w_{I,N}^{t} \eta_{N}}}\right).
    \end{equation}
    \end{theorem}

    \begin{proof}
        See Appendix~\ref{appendix:theorem1}.
    \end{proof}

\subsection{Local Optimum SC Search}\label{ss:LOS}

Given the solution to IM, determining the SC of all BSs for all slices in Problem~P1 remains a combinatorial optimization problem. To solve this problem, we propose an unsupervised-learning-assisted approach. The basic idea is to first iteratively find a locally optimal solution to SCM and then use a deep unsupervised learning technique to refine the locally optimal solution obtained by the iterative algorithm. We detail the designed iterative algorithm and the unsupervised-learning-assisted algorithm in this subsection and Subsection~\ref{ss:ULS}, respectively. 

    \begin{algorithm}[t]
         \caption{LOSCS Algorithm}\label{alg:LOS}
         \LinesNumbered
         \SetKwInOut{Input}{Input}
        \textbf{Input:} $\mathbf{W}$\\
        \textbf{Initialize:} Randomly select $m \in \mathcal{M} \backslash \{ 0 \}$, and set $\mathcal{M}^\text{s} = \mathcal{M} \backslash \{ 0 \}$, $\mathbf{l}^\mathrm{f}$, $\mathbf{l}^\mathrm{r}$, $\mathbf{a}$; \\
        Obtain $\mathbf{p}$ by Theorem~\ref{theorem1} given $\mathbf{W}$, $\mathbf{l}^\mathrm{f}$, $\mathbf{l}^\mathrm{r}$, and $\mathbf{a}$;\\
        Calculate $\Delta (\mathbf{l}^\mathrm{f}$, $\mathbf{l}^\mathrm{r}$, $\mathbf{a}, \mathbf{p})$ given $\mathbf{W}$;\\
        \While{$\mathcal{M}^\text{s} \neq \emptyset$}
        { 
            \For{$ \hat{\mathbf{l}}_{m} \in \mathcal{S}_m$}
             {	
                Obtain $\hat{\mathbf{l}}^\mathrm{f}$, $\hat{\mathbf{l}}^\mathrm{r}$, $\hat{\mathbf{a}}$ by updating the SC of SBS~$m$ with $\hat{\mathbf{l}}_{m}$;\\ 

             	\eIf{ Constraints~\eqref{eq2p}, \eqref{eq3p}, \eqref{eq2} are not satisfied }
        		{\textbf{Continue};}
        		{	
                    Obtain $\hat{\mathbf{p}}$ by Theorem~\ref{theorem1} given $\mathbf{W}$, $\hat{\mathbf{l}}^\mathrm{f}$, $\hat{\mathbf{l}}^\mathrm{r}$, and $\hat{\mathbf{a}}$;\\
                    Calculate $\Delta' (\hat{\mathbf{l}}^\mathrm{f}$, $\hat{\mathbf{l}}^\mathrm{r}$, $\hat{\mathbf{a}}, \hat{\mathbf{p}})$ given $\mathbf{W}$;\\

        			\eIf{$ \Delta' > \Delta$}
        			{$\Delta$ $\leftarrow$ $\Delta'$;\\ 
        			$\mathbf{l}_\mathrm{f}$, $\mathbf{l}_\mathrm{r}$, $\mathbf{a}$, $\mathbf{p}$ $\leftarrow$  $\hat{\mathbf{l}}_\mathrm{f}$, $\hat{\mathbf{l}}_\mathrm{r}$, $\hat{\mathbf{a}}$, $\hat{\mathbf{p}}$;\\
                    $\mathcal{M}^\text{s}$ $\leftarrow$ $\mathcal{M} \backslash \{ 0\}$;\\ 
        			}
        			{\textbf{Continue};}
        		}      	
             }
             $\mathcal{M}^\text{s}$ $\leftarrow$  $\mathcal{M}^\text{s} \backslash \{ m \}$;\\
             Randomly select $ m \in \mathcal{M}^\text{s}$;\\
         }
        \textbf{Output:} $\mathbf{l}_\mathrm{f}$, $\mathbf{l}_\mathrm{r}$, $\mathbf{a}$, $\mathbf{p}$, and $\Delta$

    \end{algorithm}

We present the local optimum SC search (LOSCS) algorithm, which iteratively updates the SC of each SBS, searching one SBS at a time, until no further energy efficiency improvement can be achieved by updating the SC of any SBS. Denote the objective function in Problem~P1 and the value of the objective function by $\Delta (\mathbf{l}^\mathrm{f}, \mathbf{l}^\mathrm{r}, \mathbf{a}, \mathbf{p})$ and $\Delta$, respectively. The algorithm is detailed in Algorithm~\ref{alg:LOS}. Let set~$\mathcal{M}^\text{s} \subset \mathcal{M}$ include the SBSs that have not been involved in the iterative search yet. Line~2 initializes set~$\mathcal{M}^\text{s} = \mathcal{M} \backslash \{ 0 \}$ and the SC of all SBSs, i.e., $\mathbf{l}^\mathrm{f}$, $\mathbf{l}^\mathrm{r}$, and $\mathbf{a}$, and randomly selects an SBS, i.e., SBS~$m$, to start searching. Given the initialized SC of SBSs, line~3 and line~4 obtain the optimal solution of IM, i.e., $\mathbf{p}$, and the corresponding value of the objective function in Problem~P1, i.e., $\Delta$. Line~5 to Line~21 search SCM solution for an SBS, corresponding to one iteration. Denote the SC of SBS~$m$ for the slices by vector $\mathbf{l}_{m} = [ l_{m,n} ]_{\forall n \in \mathcal{N}}$ which can be obtained by~\eqref{eq1}. We introduce $\mathcal S_{m}$ to represent the set that includes all possible combinations of the SC of SBS~$m$ for all slices, i.e., all possible values of vector $\mathbf{l}_{m}$ when they satisfy constraints~(\ref{p1}c) and (\ref{p1}e). During each iteration, we only search the SC of SBS~$m$ for all slices from set~$\mathcal S_{m}$ while keeping the SC of other SBSs fixed. If an SC combination yielding a larger value of $\Delta$, is found, the currently best SCM solution is updated, and set~$\mathcal{M}^\text{s}$ will be reset to the set of all SBSs; Otherwise, no change will be made. At the end of an iteration, another SBS is randomly selected from set~$\mathcal{M}^\text{s}$ for the next iteration, and the set~$\mathcal{M}^\text{s}$ is updated. All iterations stop if the set~$\mathcal{M}^\text{s}$ is an empty set, which means that a solution with a larger value of~$\Delta$ cannot be found by adjusting the SC of any SBS. The output of Algorithm~\ref{alg:LOS} is an SCM solution with the corresponding optimal solution of IM given by~\eqref{eq12}.

The computation complexity of Algorithm~\ref{alg:LOS} is $\mathcal{O}((L_\text{max}^{2}-L_\text{max})^{N} 2^{(M-1)N} I^{3}N^{3}T)$, where the computation complexity of IM in each time interval is $\mathcal{O}(I^{3}N^{3})$, and the computation complexity of SCM in each planning window is~$\mathcal{O}((L_\text{max}^{2}-L_\text{max})^{N} 2^{(M-1)N})$. Since the performance of the SCM solution found by Algorithm~\ref{alg:LOS} depends on the initial settings, we design an unsupervised-learning-assisted SC search (ULSCS) algorithm next to reduce the computation complexity of planning-stage resource management while enhancing the performance of the LOSCS algorithm by finding proper initial settings.

\subsection{Unsupervised-learning-assisted SC Search}\label{ss:ULS}

    \begin{figure}
        \begin{centering}
            \includegraphics[width=0.48\textwidth]{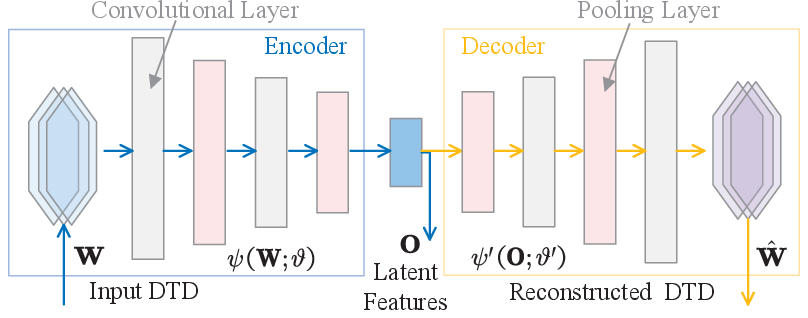}
            \par\end{centering}
        \centering{\caption{The designed DNN architecture of the auto-encoder.}\label{f:NN}}
    \end{figure} 

In each planning window, the SCM solution is related to the spatiotemporal service demands of all slices. The amount of downlink data traffic in each grid is continuous, whereas variables of SCM are discrete. As a result, similar~$\mathbf{W}$ in different planning windows may lead to the same optimal SCM solution. Thus, we propose a data-driven approach to utilize historical solutions for refining the SCM solution obtained by Algorithm~\ref{alg:LOS} in each planning window. The proposed approach consists of two components: feature extraction and solution refinement. First, we leverage an auto-encoder, a deep unsupervised learning technique, to extract the implicit and low-dimensional features of~$\mathbf{W}$ in a planning window. Second, by comparing the extracted features of~$\mathbf{W}$ in the historical and the subsequent planning window, we select some historical solutions to use as the initial settings of Algorithm~\ref{alg:LOS}. The network energy efficiency, i.e., $\Delta$, is non-decreasing over the iterations of Algorithm~\ref{alg:LOS}. As a result, choosing a historical SCM solution as the initial settings results in a relatively high performance compared to Algorithm~\ref{alg:LOS}, and the worst-case network energy efficiency equals that obtained by Algorithm~\ref{alg:LOS}.   

\subsubsection{Feature Extraction} 

Considering that the value of~$\mathbf{W}$ may vary across planning windows, we name the matrix~$\mathbf{W}$ in a planning window as a \emph{DTD instance}. The selection of a solution from a historical planning window is based on whether or not the DTD instance in the historical planning window is similar to that in the upcoming planning window. However, due to the high dimensionality of DTD instances, comparing every element in the two DTD instances is time-consuming. Therefore, reducing the dimensionality of DTD instances while retaining their essential information is important to the comparison. We utilize the deep auto-encoder technique to obtain a low-dimension representation of a DTD instance, named~\emph{latent features}. Fig.~\ref{f:NN} shows our design of deep neural networks (DNNs) for implementing the auto-encoder. The DNNs include two main parts: an encoder and a decoder. The encoder is a non-linear mapping function from a high dimensional space to a low dimensional space, i.e., extracting latent features from a DTD instance, and the decoder is a non-linear mapping function from a low dimensional space to a high dimensional space, i.e., reconstructing a DTD instance based on the latent features. Both parts are implemented by DNNs, and the DNN architecture of the decoder mirrors that of the encoder. In the training phase, the DNNs of both the encoder and the decoder are trained with the goal of minimizing the difference between the input and the reconstructed DTD instances. In the inference phase, only the DNN of the encoder is used for feature extraction~\cite{ghasedi2017deep}.

Denote the extracted latent features from a DTD instance by $\mathbf{O}$, and the sets of all possible values of $\mathbf{W}$ and $\mathbf{O}$ by~$\mathcal{W}$ and~$\mathcal{O}$, respectively. We define the encoder as the function~$\psi : \mathcal W  \rightarrow \mathcal O$ and the decoder as the function~$\psi' :\mathcal O  \rightarrow \mathcal W$. Let vectors~$\boldsymbol{\vartheta}$ and $\boldsymbol{\vartheta}'$ denote the parameters of DNNs of the encoder and the decoder, respectively. According to the designed DNN architecture for the auto-encoder, $\mathbf{W}$ and~$\mathbf{O}$ satisfy the following relations: $\mathbf{O} = \psi(\mathbf{W}; \boldsymbol{\vartheta})$ and $\hat{\mathbf{W}} = \psi'(\mathbf{O}; \boldsymbol{\vartheta}')$, where $\hat{\mathbf{W}}$ denotes the reconstructed DTD instance from the latent features~$\mathbf{O}$. To extract the latent features without neglecting useful information, the input and the reconstructed DTD instances should be as similar as possible. Therefore, the optimal values of parameters~$\boldsymbol{\vartheta}$ and~$\boldsymbol{\vartheta}'$, denoted by~$\boldsymbol{\vartheta}_{*}$ and~${\boldsymbol{\vartheta}}'_{*}$, are obtained by the following equation:
    \begin{equation}\label{eq17}
        \begin{aligned}
            \{\boldsymbol{\vartheta}_{*}, {\boldsymbol{\vartheta}}'_{*} \} & = \arg \min_{\{ \boldsymbol{\vartheta}, \boldsymbol{\vartheta}'\}} F (\mathbf{W}, \hat{\mathbf{W}}) \\
             & = \arg \min_{\{\boldsymbol{\vartheta}, \boldsymbol{\vartheta}'\}} F (\mathbf{W}, \psi'(\mathbf{O}; \boldsymbol{\vartheta}')) \\
             & = \arg \min_{\{\boldsymbol{\vartheta}, \boldsymbol{\vartheta}'\}} F (\mathbf{W}, \psi'(\psi(\mathbf{W}; \boldsymbol{\vartheta}); \boldsymbol{\vartheta}')),
        \end{aligned} 
    \end{equation}
where $F (\mathbf{W}, \hat{\mathbf{W}})$ is the cross-entropy loss function~\cite{ghasedi2017deep}. The optimal values of parameters, i.e.,~$\boldsymbol{\vartheta}_{*}$ and~${\boldsymbol{\vartheta}'}_{*}$ are obtained by using the gradient descent method to minimize the loss function~$F (\mathbf{W}, \hat{\mathbf{W}})$. The data regarding DTD instances in the set~$\Upsilon$ are utilized to train the DNNs and obtain the optimal parameters offline.         

\subsubsection{Solution Refinement} 

Using the extracted latent features of DTD instances, we define the similarity of two DTD instances in different planning windows, i.e.,~$\mathbf{W}$ and $\mathbf{W}'$, as follows:
    \begin{equation}\label{eq18}
        \begin{aligned}
            D(\mathbf{W}, \mathbf{W}') & = \frac{\psi(\mathbf{W} ; \boldsymbol{\vartheta}_{*}) \psi(\mathbf{W}'; \boldsymbol{\vartheta}_{*})}{\| \psi(\mathbf{W}; \boldsymbol{\vartheta}_{*}) \| \| \psi(\mathbf{W}'; \boldsymbol{\vartheta}_{*}) \|} \\
            & = \frac{\mathbf{O} \cdot \mathbf{O}'}{\| \mathbf{O} \| \| \mathbf{O}' \|},
        \end{aligned}
    \end{equation}
where $\mathbf{O}= \psi(\mathbf{W}; \boldsymbol{\vartheta}_{*})$ and $\mathbf{O}' = \psi(\mathbf{W}'; \boldsymbol{\vartheta}_{*})$ denote the latent features of DTD instances~$\mathbf{W}$ and~$\mathbf{W}'$ given the well-trained DNN of the encoder with parameter~$\boldsymbol{\vartheta}_{*}$, respectively.  

    \begin{algorithm}[t]
        \caption{ULSCS Algorithm}\label{alg:ULS}
        \LinesNumbered
        \SetKwInOut{Input}{Input}
        \textbf{Input:} $\boldsymbol{\vartheta}_{*}$, $\mathbf{W}$, $|\Upsilon^\text{re}|$, and $\Upsilon$ \\
        % \textbf{Initialize:} \\
        Calculate the similarity between $\mathbf{W}$ and each DTD instance, i.e., $\mathbf{W}'$, contained in~$\Upsilon$ by \eqref{eq18};\\
        $\Upsilon^\text{re}$ $\leftarrow$ Select data records containing the $|\Upsilon^\text{re}|$ most similar DTD instances from set~$\Upsilon$;\\
        Obtain $\Delta$, $\mathbf{l}_\mathrm{f}$, $\mathbf{l}_\mathrm{r}$, $\mathbf{a}$, $\mathbf{p}$ by Algorithm~\ref{alg:LOS}, given $\mathbf{W}$;\\
        \For{$ \upsilon \in \Upsilon^\mathbf{re} $}
        {   
            Obtain $\mathbf{l}_\mathrm{f}^\text{re}$, $\mathbf{l}_\mathrm{r}^\text{re}$, $\mathbf{a}^\text{re}$ from data record~$\upsilon$;\\

            Obtain $\Delta'$, $\mathbf{l}_\mathrm{f}'$, $\mathbf{l}_\mathrm{r}'$, $\mathbf{a}'$, $\mathbf{p}'$ by Algorithm~\ref{alg:LOS} given $\mathbf{W}$ and the initial settings of $\mathbf{l}_\mathrm{f}^\text{re}$, $\mathbf{l}_\mathrm{r}^\text{re}$, and $\mathbf{a}^\text{re}$;\\ 
            \eIf{$ \Delta' > \Delta$}
                {$\Delta$, $\mathbf{l}_\mathrm{f}$, $\mathbf{l}_\mathrm{r}$, $\mathbf{a}$, $\mathbf{p}$  $\leftarrow$ $\Delta'$, $\mathbf{l}_\mathrm{f}'$, $\mathbf{l}_\mathrm{r}'$, $\mathbf{a}'$, $\mathbf{p}'$;\\}
                {\textbf{Continue};}

        }
        Create a data record $\upsilon'$ containing $\mathbf{W}$, $\mathbf{l}_\mathrm{f}$, $\mathbf{l}_\mathrm{r}$, $\mathbf{a}$, and $\mathbf{p}$;\\
        Add $\upsilon'$ to $\Upsilon$;\\ 
        \textbf{Output:} $\mathbf{l}_\mathrm{f}$, $\mathbf{l}_\mathrm{r}$, $\mathbf{a}$, $\mathbf{p}$,  and $\Delta$
    \end{algorithm} 

Algorithm~\ref{alg:ULS} presents the procedure for refining the solutions obtained by Algorithm~\ref{alg:LOS}. We refer to the collection of information on the DTD instance, i.e.,~$\mathbf{W}$, and the corresponding solution obtained by Algorithm~\ref{alg:LOS}, i.e.,~$\mathbf{l}_\mathrm{f}$, $\mathbf{l}_\mathrm{r}$, $\mathbf{a}$, and $\mathbf{p}$, in a planning window as a data record, denoted by $\upsilon$. Denote the set of data records and the number of data records in the set by~$\Upsilon$ and~$|\Upsilon|$, respectively. The value of~$|\Upsilon|$ can be determined by balancing the computation complexity and the performance of the ULSCS algorithm. Using~\eqref{eq18}, Line~2 calculates the similarity between the DTD instance in the upcoming planning window, i.e.,~$\mathbf{W}$, and each DTD instance in the set~$\Upsilon$. Based on the calculated similarities, a set of data records containing the~$|\Upsilon^\text{re}|$ most similar DTD instances, denoted by $\Upsilon^\text{re} \subseteq \Upsilon$, is selected. Line~4 obtains the solution to Problem~P1, i.e.,~$\mathbf{l}_\mathrm{f}$, $\mathbf{l}_\mathrm{r}$, $\mathbf{a}$, and $\mathbf{p}$, and the corresponding performance~$\Delta$ by calling Algorithm~\ref{alg:LOS}. From Lines~6 to~12, each historical SCM solution in the set~$\Upsilon^\text{re}$, i.e., $\mathbf{l}_\mathrm{f}^\text{re}$, $\mathbf{L}_\mathrm{r}^\text{re}$, and $\mathbf{a}^\text{re}$, is used in the initialization step (Line 2) of Algorithm~\ref{alg:LOS}, and the corresponding performance~$\Delta'$ and solution~$\mathbf{l}_\mathrm{f}'$, $\mathbf{l}_\mathrm{r}'$, $\mathbf{a}'$, $\mathbf{p}'$ are obtained. If~$\Delta' > \Delta$, the solution to Problem~P1 is updated as~$\mathbf{l}_\mathrm{f}'$, $\mathbf{l}_\mathrm{r}'$, $\mathbf{a}'$, $\mathbf{p}'$; Otherwise, the solution to Problem~P1 remains $\mathbf{l}_\mathrm{f}$, $\mathbf{l}_\mathrm{r}$, $\mathbf{a}$, and $\mathbf{p}$. As a result, the performance of Algorithm~\ref{alg:ULS} is either better than or equal to that of Algorithm~\ref{alg:LOS}. When all historical SCM solutions in the set~$\Upsilon^\text{re}$ have been utilized, lines~14 and~15 create a new data record containing the DTD instances and the corresponding solution, i.e.,~$\mathbf{l}_\mathrm{f}$, $\mathbf{l}_\mathrm{r}$, $\mathbf{a}$, and $\mathbf{p}$, and add the data record to the set~$\Upsilon$, which can be useful in subsequent planning windows.

By using deep unsupervised learning, Algorithm~\ref{alg:ULS} can reduce the computation complexity of planning-stage resource management of Algorithm~\ref{alg:LOS} when the set~$\Upsilon$ contains extensive historical data records. The computation complexity of the Algorithm~\ref{alg:ULS} is $\mathcal{O}(|\Upsilon| O X I^{3}N^{3})$ for selecting the best solution to~Problem~P2 from set~$\Upsilon$, where $O$ represents the dimensionality of latent feature~$\mathbf{O}$, $X = \sum_{j=1}^{J-1}{B_{j}B_{j+1}}$ denotes the computation complexity of the inference of the encoder (i.e., DNN~$\psi$) with~$J$ layers, and~$B_{j}$ represents the number of neurons in layer~$j$. Similar to the scheme used for experience replay in reinforcement learning~\cite{zhou2020deep,liu2020deepslicing}, we fix the maximum number of data records in the set~$\Upsilon$, i.e.,~$|\Upsilon|$, and keep the newly collected data records in~$\Upsilon$. As a result, by collecting and using new data records, Algorithm~\ref{alg:ULS} can enhance the performance of Algorithm~\ref{alg:LOS} while avoiding high computation complexity.

\section{Performance Evaluation}

In this section, we first introduce the simulation settings. Then, we evaluate the performance of the proposed RAN slicing framework with the proposed AI-assisted approach.

\subsection{Simulation Settings}

The maximum SC and the antenna height of all SBSs are set to identical. The SC radius of the MBS and the maximum SC radius of each SBS are set to~1,500\,m and 850\,m, respectively. The carrier frequency of each BS is set to 1,500 MHz. The total available bandwidth of each BS and the sub-carrier spacing are set to 100\,MHz and 30\,kHz, respectively. Based on the COST 231-Hata Model in 3GPP standard~\cite{3GPP43030}, the average channel gain of downlink transmission within grid~$i$ for slice~$n$ in time interval~$t$, i.e.,~$h_{i,n,i',n'}^{t}$, is approximated as the following equation:
    \begin{equation}\label{eq19}
       \begin{aligned} 
            h_{i,n}^{t} & =  46.55 + 33.81 \times \log(f^\text{c}_{m}) - 13.82 \times \log(H_{m}) + \\
            & ((44.9 - 6.55 \times \log(H_{m})) \times \log(d_{m_{i,n},i}),
       \end{aligned} 
    \end{equation}
where~$d_{m_{i,n},i}$ is the distance (in kilometers) between BS~$m_{i,n}$ and the center of grid~$i$, $f^\text{c}_{m}$ is the carrier frequency (in MHz) of BS~$m$,~$H_{m}$ is the antenna height (in meters) of BS~$m$, and~$H_\text{MBS}$ and~$H_\text{SBS}$ represent the antenna heights of the MBS and each SBS, respectively. UTs within the network coverage area in a time interval are distributed according to a Poisson point distribution (PPP). The rates of the PPP are the same across all time intervals within each planning window but different across planning windows. For each UT, its downlink data traffic load follows a Poisson process during each planning window. The mean values of downlink data traffic loads are different among UTs. We randomize the mean downlink data traffic load for each UT during a planning window within the interval of $[0.1, 1.5]$\,Mbits. Other simulation parameters are listed in Table~II.

    \begin{table}[t]
        \normalsize
        \centering
        \captionsetup{justification=centering,singlelinecheck=false}
        \caption{Simulation Parameters}\label{table2}
        \begin{tabular}{c|c|c|c}
            \hline\hline
             Parameter & Value & Parameter & Value\\
             \hline\hline
             $N$ & 2 & T & 3 \\
             \hline
            $\left[\gamma_{1}^\mathrm{min}, \gamma_{2}^\mathrm{min} \right]$ & [7, 11] dB  & $\left[\lambda_1, \lambda_2\right]$ & [1, 1]\\
             \hline
             $H_\text{MBS}$ & 50\,m  & $H_\text{SBS}$ & 15\,m \\
             \hline
             $\rho$ & 1 & $N_0$ & -174\,dBm/Hz \\
             \hline
        \end{tabular}
    \end{table}

The implementation of the DNNs for the auto-encoder is as follows. The DNN of the encoder contains 3 convolutional layers with channel sizes of 32, 64, and 128 respectively. The kernel size is set as (3, 3) for both convolutional layers, respectively. Each convolutional layer is followed by a max-pooling layer with pool size (2, 2). Two fully-connected layers are then added with 512 and 64 neurons, followed by the output layer. The DNN architecture of the decoder is the reverse of that of the encoder. We adopt the Adam optimizer to train the DNNs. There are 8,000 different DTD instances used for the DNN training.

We compare the proposed RAN slicing framework with the following two benchmark schemes for IM and SCM, respectively:
    \begin{itemize}
        \item \emph{Cell-based IM:} The downlink transmission power of each BS is the same for all grids within the SC of the BS;
        \item \emph{Cell zooming (CZ):} The SC of each SBS is the same for all slices.
    \end{itemize}

\subsection{Performance of Grid-based IM}

In this subsection, we investigate the performance of the proposed grid-based IM in a simple network scenario with 1~MBS, 1~SBS, and 1~slice. 

    \begin{figure}[t]
      \centering
        \subfigure[Average downlink transmission power of the SBS and the MBS versus SC radius.]
        {\includegraphics[width=0.45\textwidth]{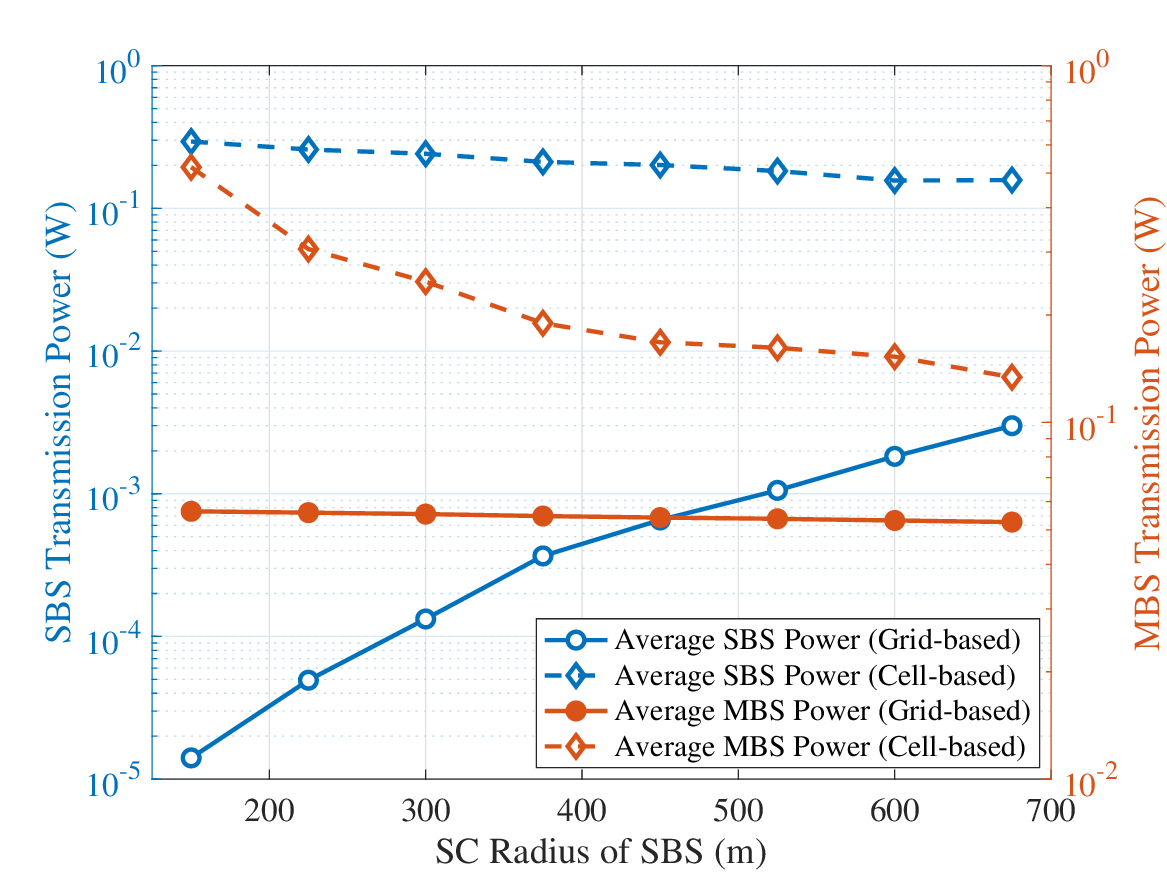}}
    	\subfigure[Total downlink transmission power and network energy efficiency versus SC radius.]
    	{\includegraphics[width=0.45\textwidth]{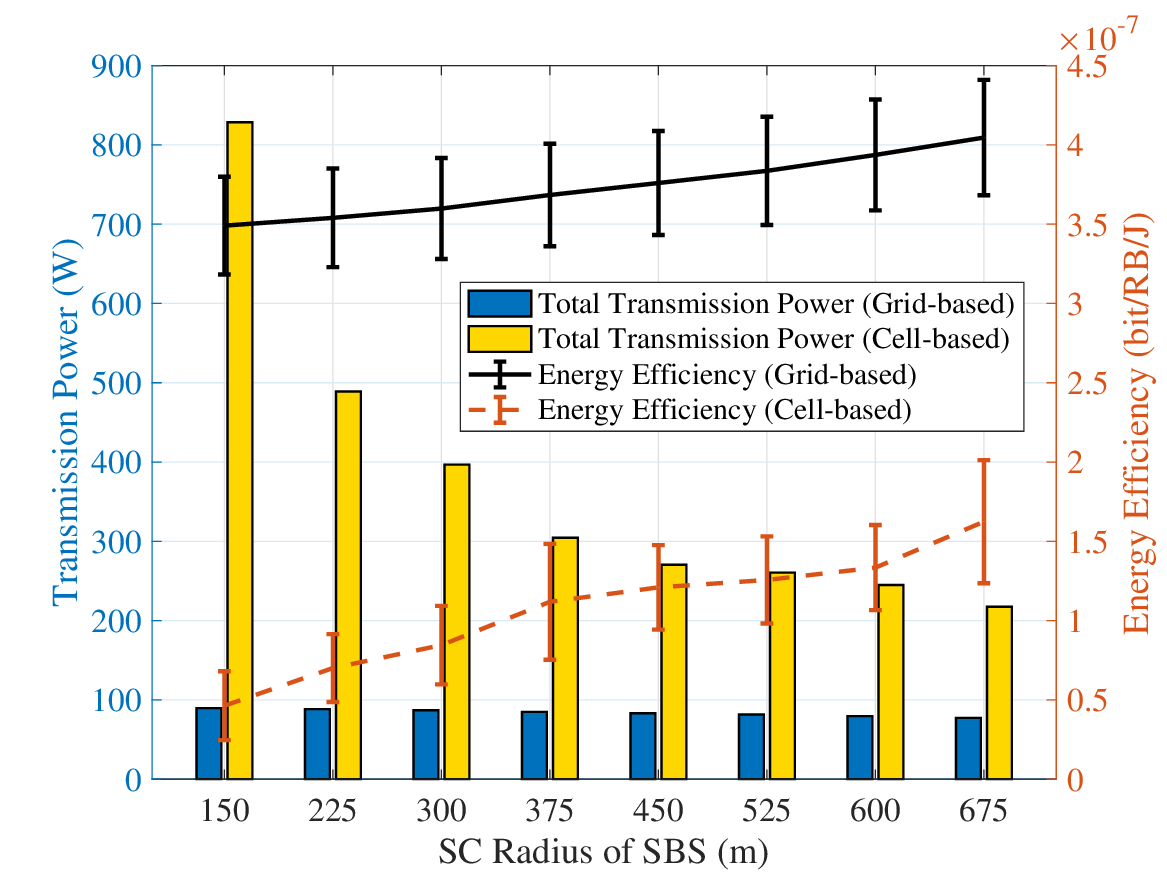}} \\
      \caption{Comparison between the proposed grid-based IM and cell-based IM.}
    	\label{r:power}
    \end{figure}

In Fig.~\ref{r:power}(a), we compare the performance of transmission power obtained by the proposed grid-based IM with that obtained by cell-based IM. We average the transmission power of all grids within the SC of each BS for comparison. To satisfy the SINR requirement of the slice, the average transmission power of the MBS decreases, and the average transmission power of the SBS increases with the SC radius of the SBS for both grid-based and cell-based IM. This is because the number of grids covered by the MBS and the SBS decreases and increases, respectively. However, the MBS and SBS can achieve lower transmission power with grid-based IM compared to cell-based IM since the proposed grid-based IM can differentiate the transmission power based on their different locations. In addition, the slopes of all curves can vary with the SC radius of the SBS. This is because the uneven spatial distribution of data traffic loads results in non-uniform increments of data traffic loads for both the SBS and the MBS.

As shown in Fig.~\ref{r:power}(b), we compare the performance of the two schemes in total transmission power and network energy efficiency. We observe that the proposed grid-based IM achieves higher network energy efficiency and lower total transmission power. The reason is that the proposed grid-based IM has a higher spatial granularity. Thus, the transmission power for each grid can be individually optimized to mitigate the interference among BSs in accordance with the DTD and the BS locations.  
 
    \begin{figure}[t]  
     \centering   
      \includegraphics[width=0.45\textwidth]{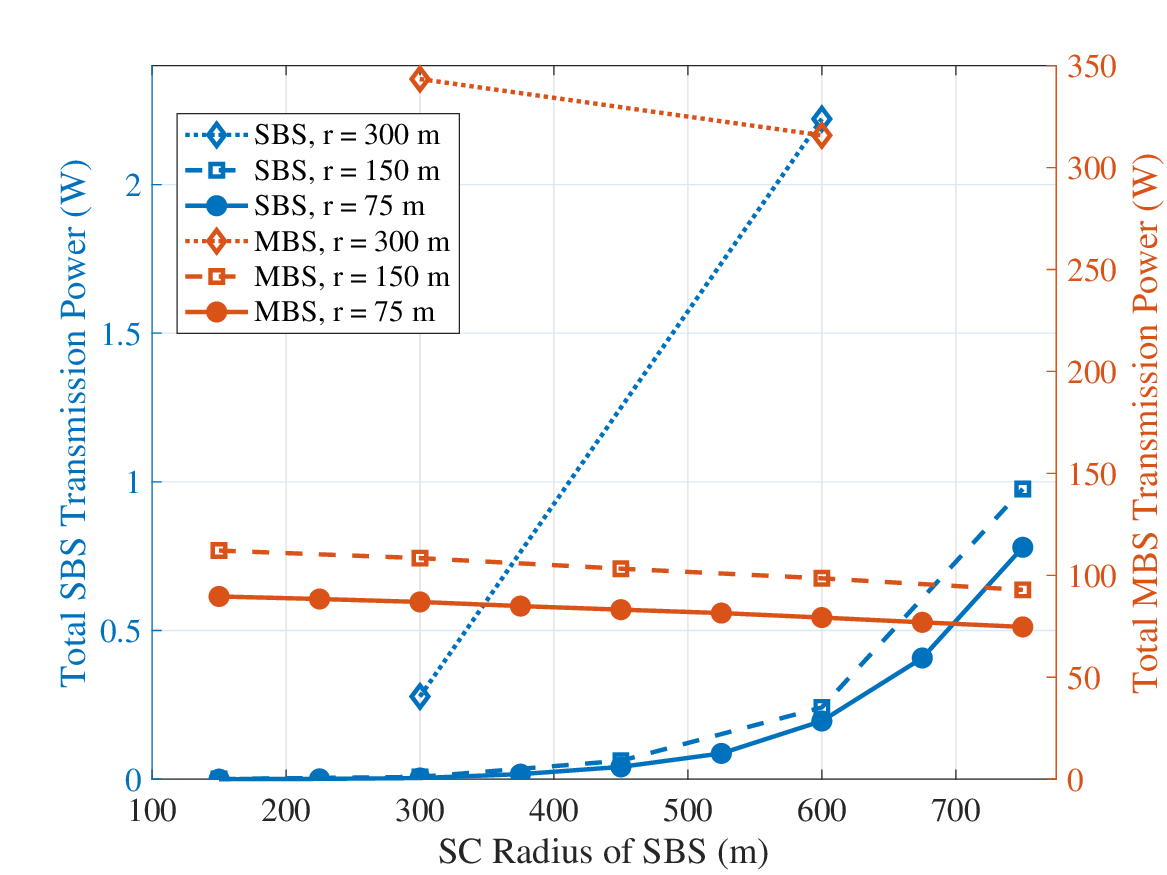}\\  
      \caption{The impact of spatial granularity on IM.}\label{r:grid}
      % \vspace{-0.9cm}
    \end{figure}

Next, we examine the impact of the spatial granularity on the network energy efficiency of grid-based IM. Fig.~\ref{r:grid} shows the total transmission power of the MBS and the SBS of grid-based IM with different grid diameters, i.e., different values of~$r$. From this figure, we can make three observations. First, similar to case in Fig.~\ref{r:power}, the total transmission power of the MBS of grid-based IM increases with the SC radius of the SBS, while the total transmission power of the SBS of grid-based IM decreases with the SC radius of the SBS. Second, with grid-based IM, the total transmission power of each BS decreases with the grid diameter. This is because, when the grid diameter is smaller, the network can be divided into more grids and IM can be more fine-grained to suit the specific DTD. Third, if the grid diameter is sufficiently small (e.g., below 150\,m), the effect of further decreasing the grid diameter on the total transmission power diminishes. This is because the total transmission power of each BS must exceed a threshold to satisfy the SINR requirement of each slice given a DTD.

\subsection{Performance of Slicing-based Resource Management}

    \begin{figure}[t]  
     \centering   
      \includegraphics[width=0.45\textwidth]{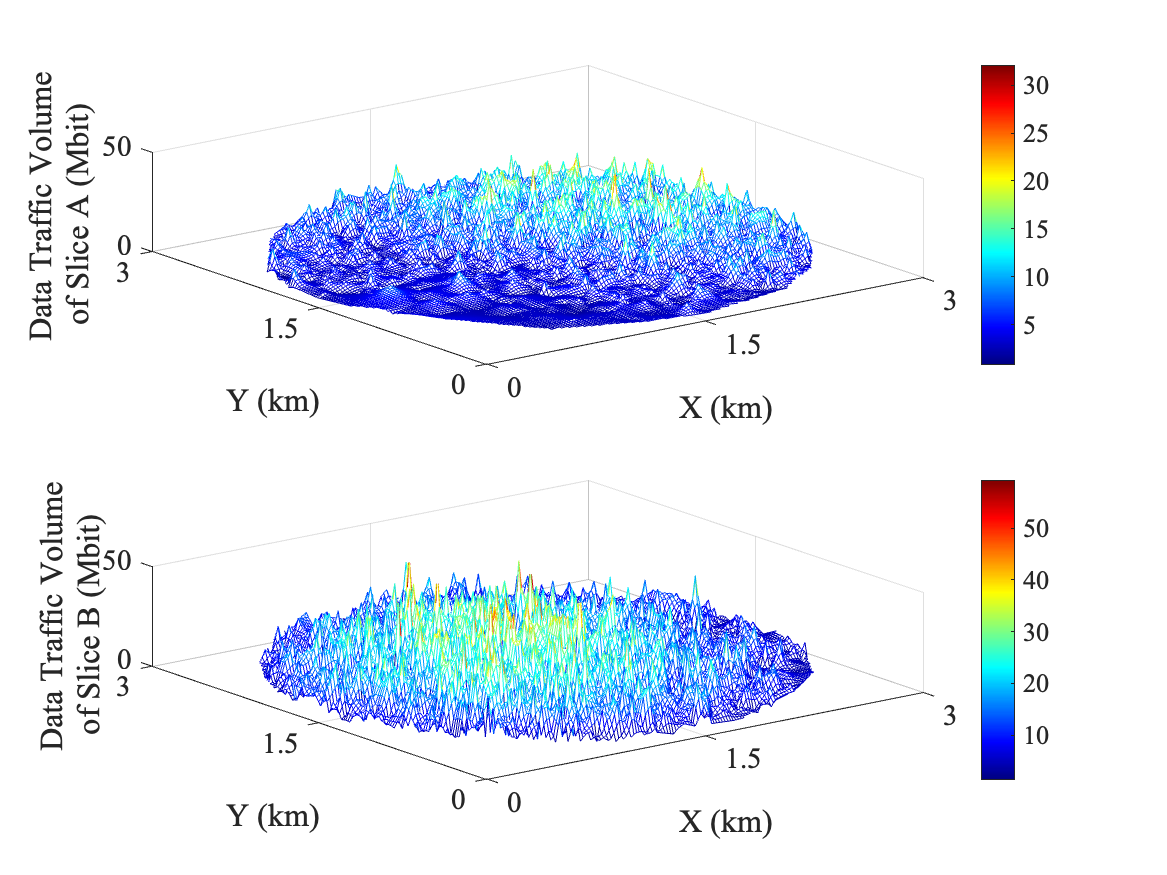}\\  
      \caption{The DTDs of two slices in a time interval.}\label{r:DTD}
    \end{figure}

In this subsection, we examine the performance of the proposed RAN slicing framework in a network scenario with 1~MBS, 1 to 8~SBSs, and 2 slices. The DTDs of the two slices are different in a planning window. The DTDs of the two slices in a time interval is shown in Fig.~\ref{r:DTD}.

Considering the network scenario with 1 MBS, 8 SBSs, and 2 slices, we compare the performance of the proposed schemes with benchmark schemes as shown in Fig.~\ref{r:traffic}. In Fig.~\ref{r:traffic}(a), we compare the network energy efficiency of the proposed flexible binary slice zooming plus grid-based IM (abbreviated as "SZ+ Grid-based IM") with that of two benchmark schemes, named ``CZ + Cell-based IM'' and ``CZ+ Grid-based IM'', averaged over 20 DTD instances. Three observations can be made from this figure. First, the network energy efficiency of all schemes increases with the number of SBSs. This is because, more SBSs can cover more grids, and the downlink transmissions within the grids from SBSs have a higher channel gain than that from the MBS, thereby improving network energy efficiency. Second, the proposed scheme outperforms the benchmark schemes in network energy efficiency in the cases with different number of SBSs. The reason is that the proposed scheme achieves fine-grained IM and SCM in time, space, and slices dimensions based on the different SINR requirements and DTDs of slices. Third, by comparing the ``CZ + Cell-based IM'' scheme with the ``CZ + Grid-based IM'' scheme, the performance advantage, i.e., the improvement (in percentage) of the proposed "SZ+ Grid-based IM" scheme compared to the ``CZ+ Grid-based IM'' scheme, increases with the number of SBSs. This is because, as more SBSs are deployed, the proposed scheme has more SC options available for selection, resulting in better interference management among BSs. Therefore, the percentage improvement compared to other schemes increases with the number of SBSs. 

In Fig.~\ref{r:traffic}(b), we show the temporal variations in network energy efficiency of the proposed scheme across multiple planning windows. The Poisson data arrival rate averaged over all UTs in each slice varies across planning windows, and, accordingly, the network energy efficiency of the proposed scheme temporally varies. Meanwhile, we can observe that the proposed scheme outperforms that of the ``CZ + Grid-based IM'' scheme in each planning window due to the high adaptivity of the proposed scheme in coping with spatiotemporal network dynamics. Fig.~\ref{r:traffic}(c) shows the cumulative distribution function of the network energy efficiency of the three schemes over the 40 different DTD instances in the same case. We can observe from Fig.~\ref{r:traffic}(c) that the proposed schemes achieve higher network energy efficiency than the benchmark schemes for most DTD instances. 

    \begin{figure*}[t]
      \centering
            \subfigure[Average network energy efficiency versus the number of SBSs.]
        {\includegraphics[width=0.31\textwidth]{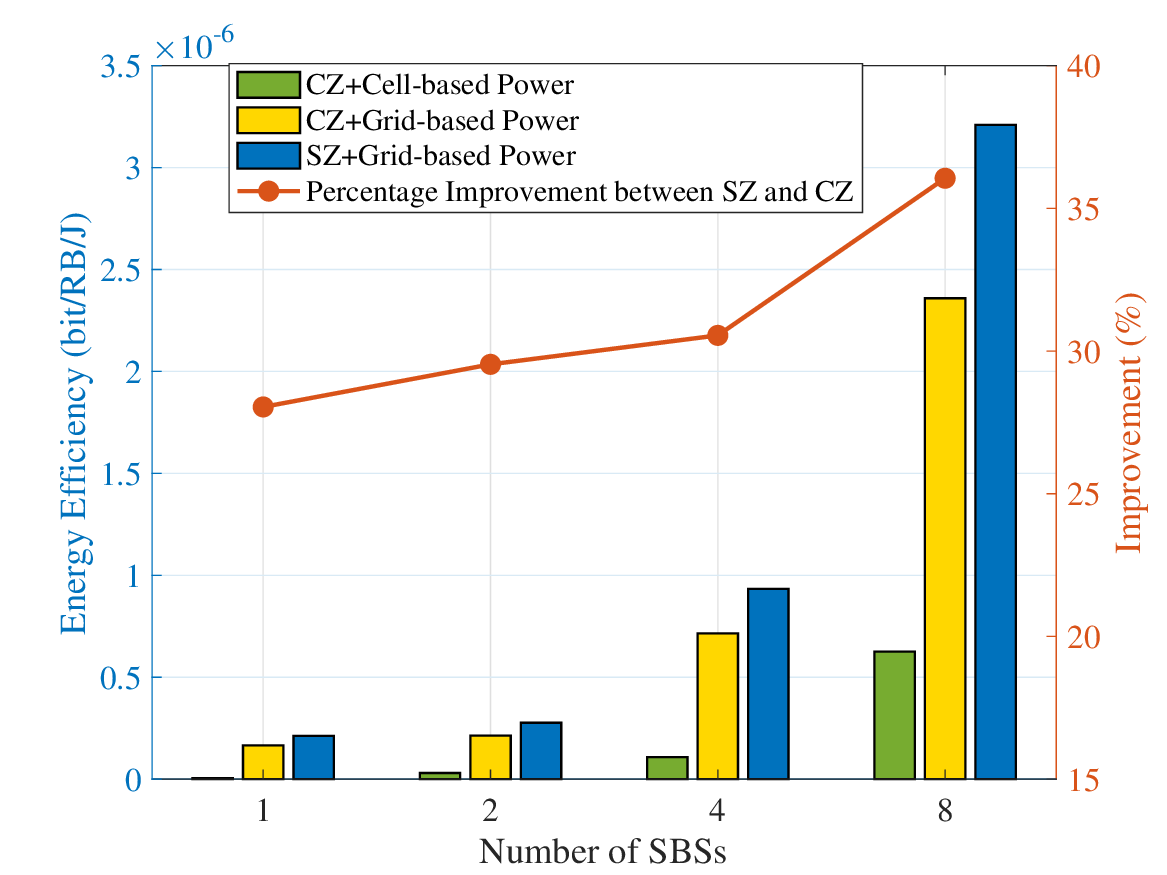}} 
            \subfigure[Temporal variations in data traffic load and network energy efficiency.]
        {\includegraphics[width=0.31\textwidth]{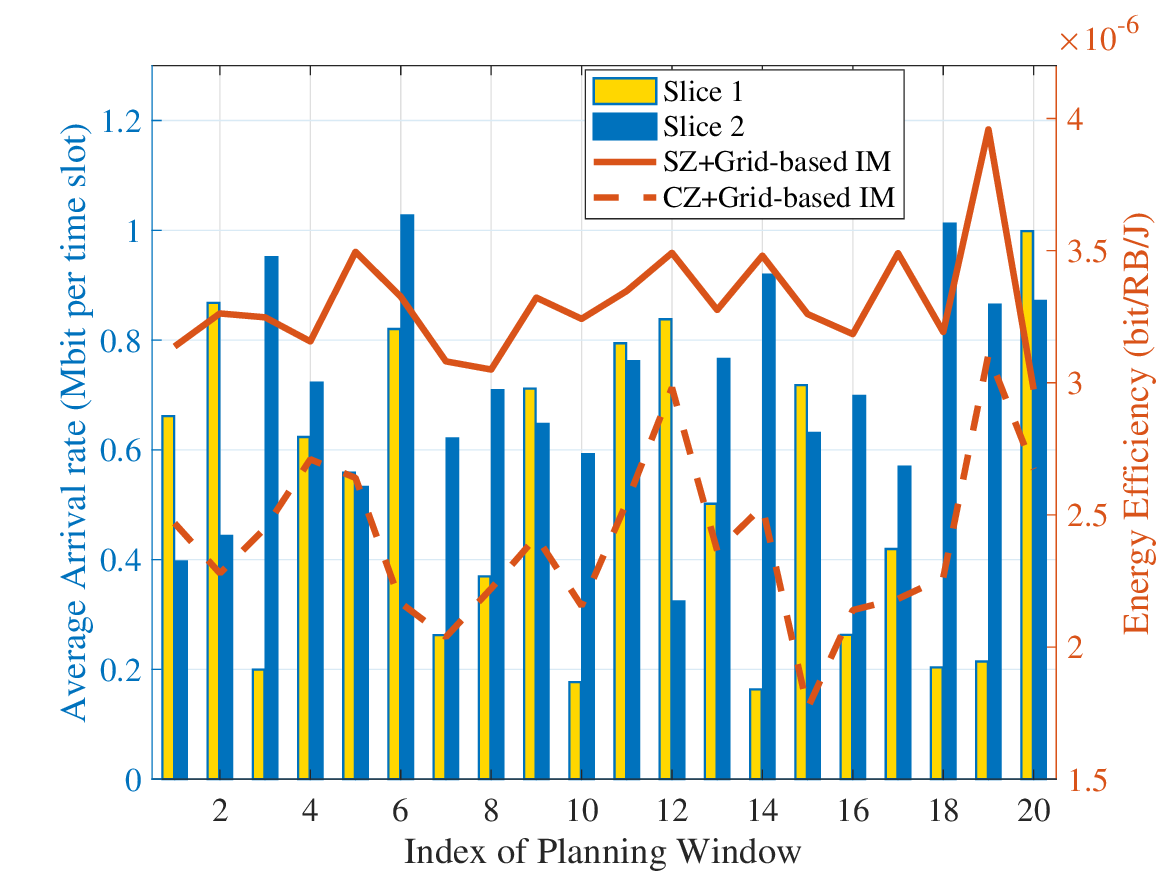}} 
            \subfigure[The cumulative distribution function of network energy efficiency.]
        {\includegraphics[width=0.31\textwidth]{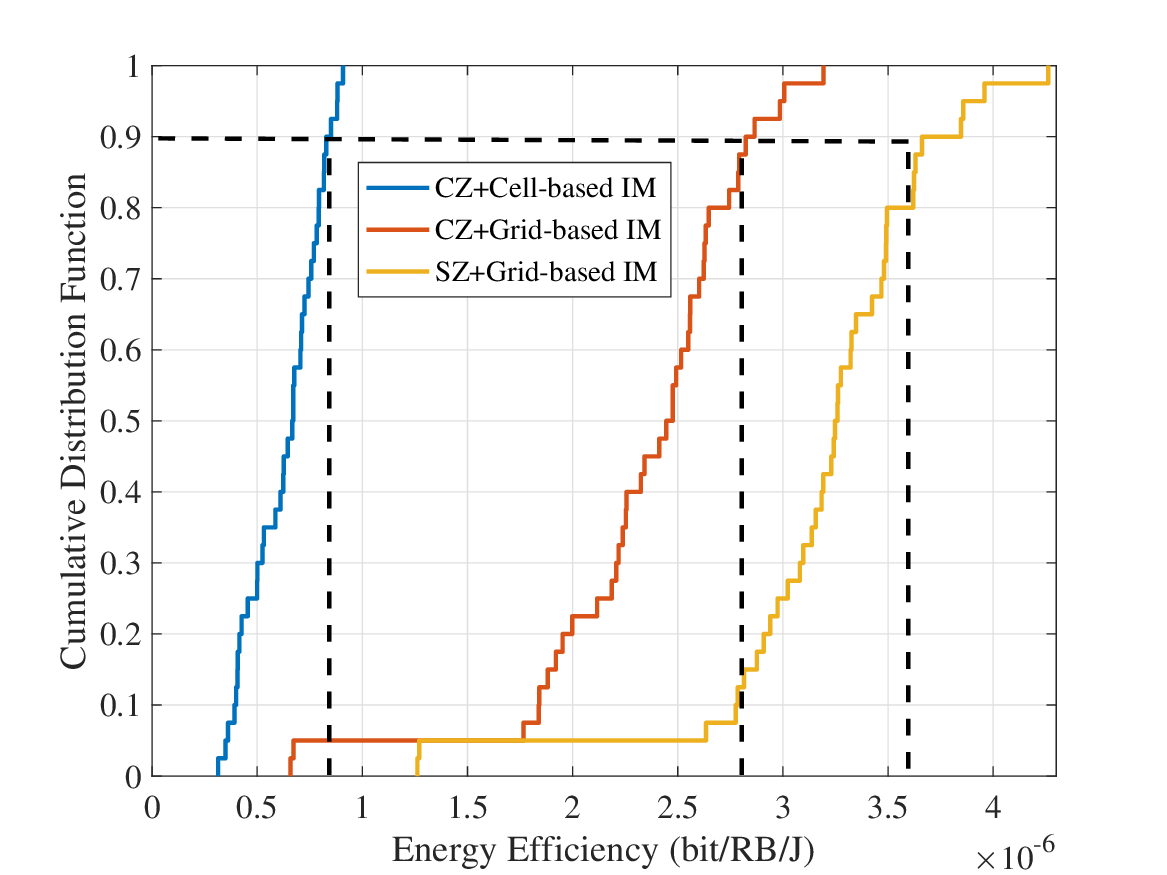}}\\
      \caption{Performance comparison between the proposed schemes and benchmark schemes.}
        \label{r:traffic}
    \end{figure*}

        \begin{figure*}[t]
      \centering
            \subfigure[Network energy efficiency versus the number of slices.]
        {\includegraphics[width=0.45\textwidth]{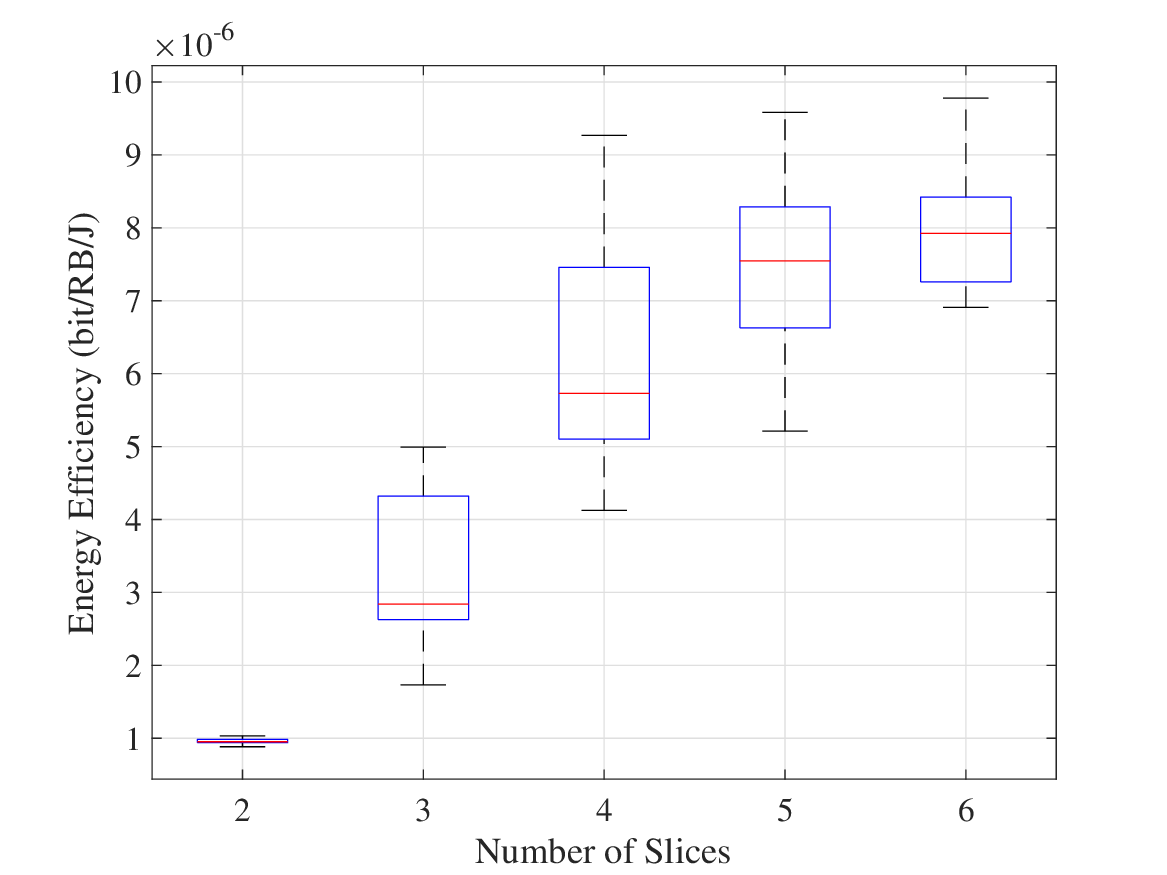}} 
            \subfigure[Network energy efficiency versus the total available bandwidth of each BS.]
        {\includegraphics[width=0.45\textwidth]{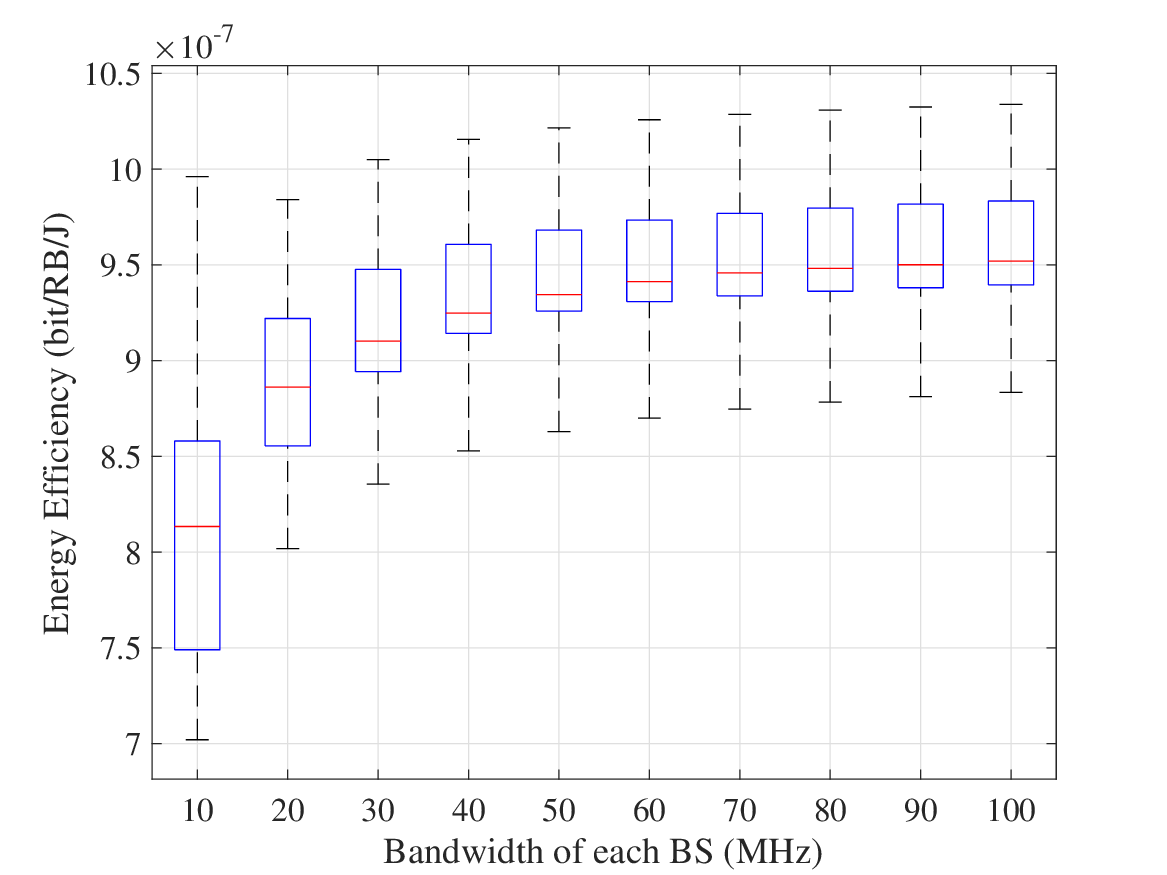}}\\
      \caption{The performance of the developed RAN slicing framework in different network scenarios.}
        \label{r:R1}
    \end{figure*}

In Fig.~\ref{r:R1}, we show the performance of the developed RAN slicing framework in different network scenarios. Considering the network with 1 MBS, 4 SBSs, and the grid diameter of 150\,m, we show the network energy efficiency versus the number of slices and the total available bandwidth of each BS in Fig.~\ref{r:R1}(a) and Fig.~\ref{r:R1}(b), respectively. A box plot representing the range of network energy efficiency over 10 independent simulation runs is shown in Fig.~\ref{r:R1}(a), in which the number of slices is set from 2 to 6, and the overall data traffic load of all slices is fixed in each simulation run. We can make the following two observations. First, the network energy efficiency of the developed scheme increases with the number of slices. This is because, for the same DTD, the number of decision variables of IM and SCM in the developed scheme increases with the number of slices, thereby improving the granularity of slicing-based resource management. As a result, the developed scheme can achieve higher energy efficiency by balancing the overall data traffic load across BSs due to the refined granularity in the slice dimension. Second, the effect of increasing the number of slices on the network energy efficiency diminishes when the number of slices increases since it becomes more difficult for IM and SCM to satisfy the SINR requirement of each slice.

In Fig.~\ref{r:R1}(b), varying the total bandwidth of each BS from 10\,MHz to 100\,MHz, we present the box plot of network energy efficiency over 10 independent simulation runs for each bandwidth setting. We can observe that network energy efficiency increases with the total available bandwidth of each BS. This is because, for the same data traffic load of each BS, increasing the total bandwidth of each BS can reduce the likeliness of the planning-stage interference among BSs (as discussed in Section IV.B). Consequently, the required transmission power to satisfy the SINR requirement of each slice is reduced, thereby improving the network energy efficiency.

\subsection{Performance of the ULSCS Algorithm}

In this subsection, we evaluate the energy efficiency performance of the proposed ULSCS algorithm and the LOSCS algorithm as well as the impact of the number of data records i.e., $|\Upsilon|$, and the number of selected data records, i.e., $|\Upsilon^\text{re}|$. We consider a network with 8~SBSs, 1~MBS, and 2 slices.

    \begin{figure}[t]  
     \centering   
      \includegraphics[width=0.45\textwidth]{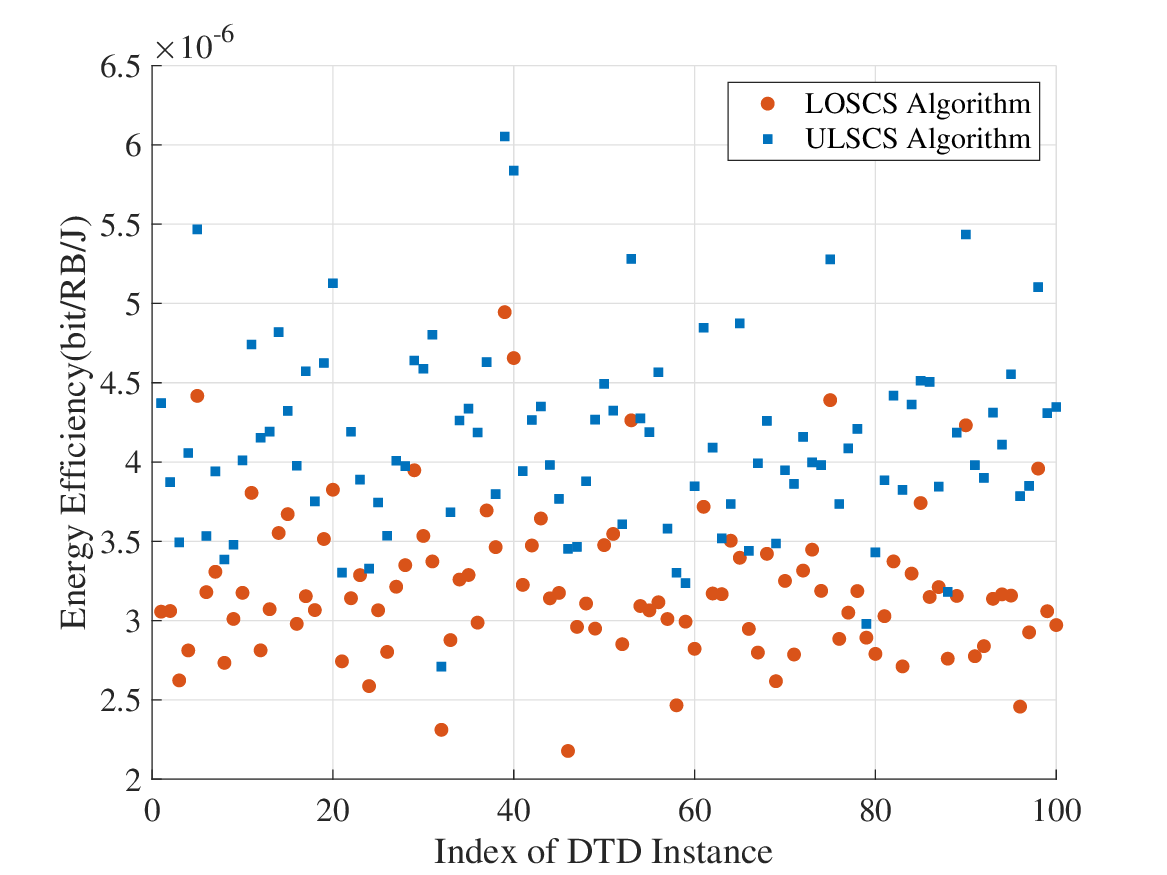}\\  
      \caption{Network energy efficiency comparison between the LOSCS and ULSCS algorithms.}\label{r:ee_data}
    \end{figure}

In Fig.~\ref{r:ee_data}, we compare the energy efficiency performance of the ULSCS and LOSCS algorithms for 100 cases with different DTD instances. The network energy efficiency achieved by the ULSCS algorithm is higher than that achieved by the LOSCS algorithm in all cases. The ULSCS algorithm selects some historical solutions to use as the initial settings of the LOSCS algorithm, which results in relatively high performance compared to the LOSCS algorithm. The worst-case network energy efficiency of the ULSCS algorithm equals that obtained by the LOSCS Algorithm.

    \begin{figure}[t]
      \centering
        \subfigure[The impact of $|\Upsilon|$.]
        {\includegraphics[width=0.45\textwidth]{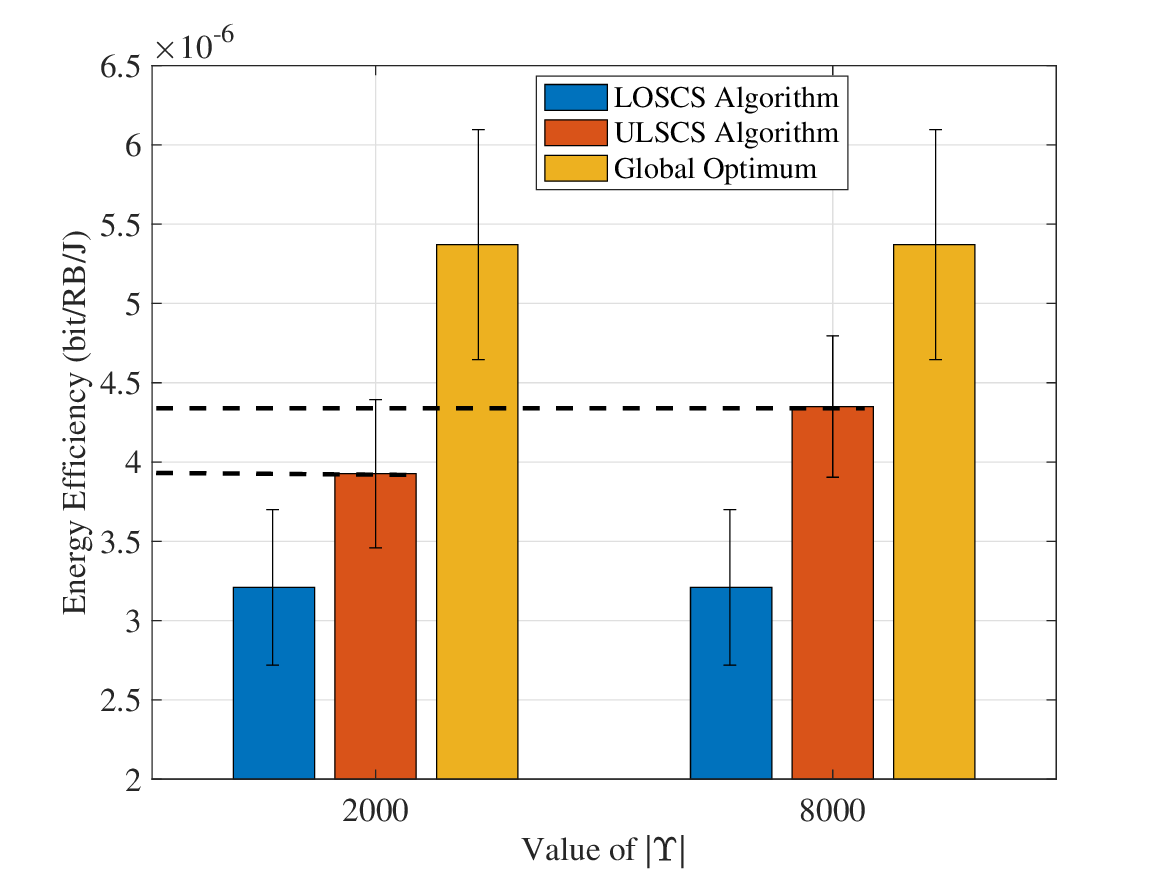}}
        \subfigure[The impact of $|\Upsilon^\text{re}|$.]
        {\includegraphics[width=0.45\textwidth]{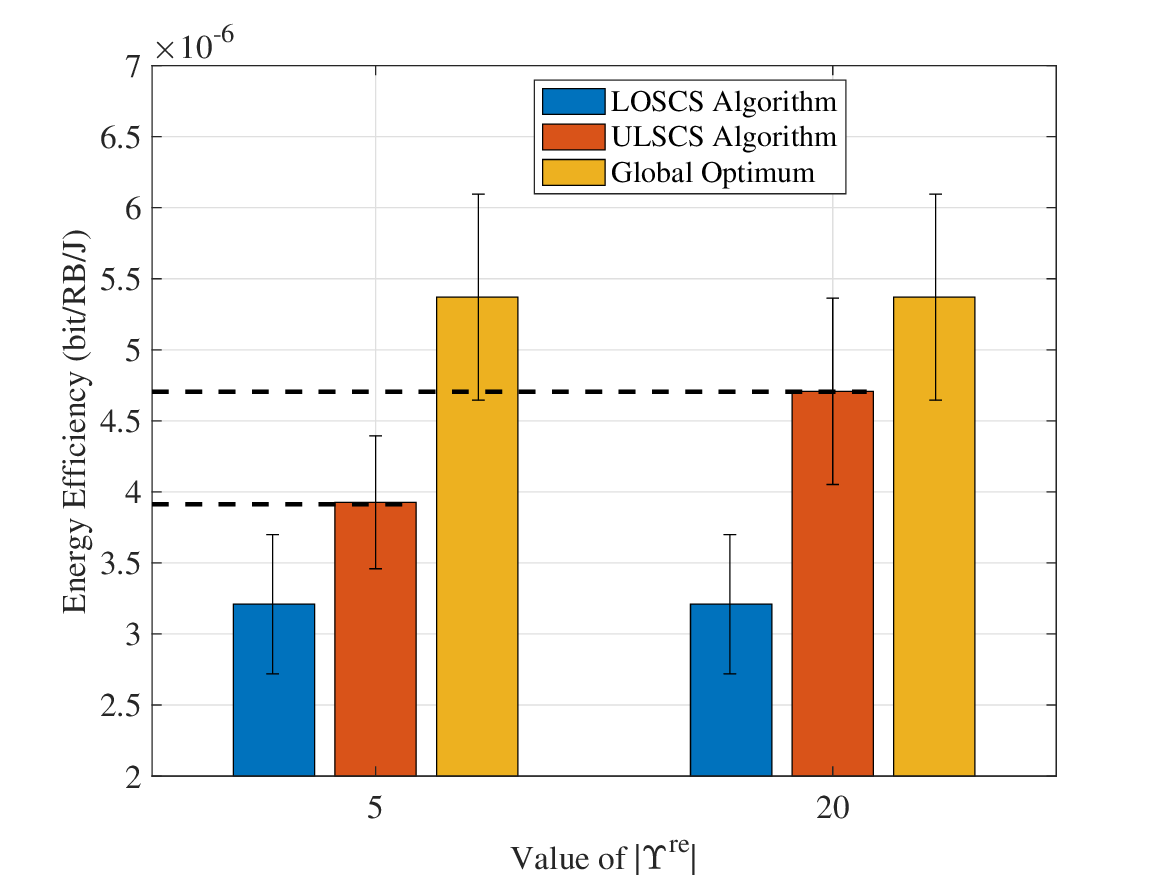}} \\
      \caption{Network energy efficiency given different values of~$|\Upsilon|$ and~$|\Upsilon^\text{re}|$, respectively.}
        \label{r:ee}
    \end{figure}

In Fig.~\ref{r:ee}(a), we evaluate the network energy efficiency of the ULSCS algorithm, averaged over 40 DTD instances, given different number of data records, i.e., different values of~$|\Upsilon|$. Two observations can be made in Fig.~\ref{r:ee}(a). First, the performance gap between the the ULSCS algorithm and the LOSCS algorithm increases when more data records are used. This is because having more data records in~$\Upsilon$ can improve the performance of DNN training and provide a large number of historical DTD instances for solution refinement. Second, the performance of the ULSCS algorithm can approach the optimum global value, especially when a large value of~$|\Upsilon|$ is used. 
Moreover, we examine the impact of the number of selected data records, i.e., different values of~$|\Upsilon^\text{re}|$, in Fig.~\ref{r:ee}(b). The performance gap between the ULSCS algorithm and the LOSCS algorithm increases with the number of selected data records, and performance of the ULSCS algorithm can approach the global optimum when a larger number of selected historical solutions are used for solution refinement. This is because more data records in~$|\Upsilon^\text{re}|$ result in more similar DTD instances being selected as the initial settings in the ULSCS algorithm, and thus benefit achieving global optimum. Consequently, Fig.~\ref{r:ee} demonstrates the potential of the AI-assisted approach to address the slicing-based resource management problems.

\section{Conclusion and Future Work}

In this paper, we have designed a RAN slicing framework for a two-tier RAN to determine the SC and transmission power of the BSs. The proposed framework introduces customized SC for different services and improves the granularity of IM to suit service demands in the spatial, temporal, and slice dimensions. Based on the framework, a network energy efficiency maximization problem has been formulated, which takes into account the inter-slice and intra-slice interference and diverse QoS requirements of slices. The proposed AI-assisted approach decouples the problem into two sub-problems and  solve them by incorporating deep unsupervised learning with optimization methods. The results have demonstrated the effectiveness of the proposed RAN slicing framework in improve energy efficiency, and the efficiency of the developed AI-assisted approach. The proposed framework and approach extend the advantages of slicing-based resource management towards supporting diverse services in RANs. In the future, we will investigate slicing-based resource management considering the coupling between the planning and operation stages.

% \section*{Acknowledgment}
% This work was supported in part by the Natural Sciences and Engineering Research Council (NSERC) of Canada, and in part by the National Natural Science Foundation of China (NSFC) under Grant No. 91638204, 62071356, and 62002389.

\appendix

\subsection{Proof of Theorem~\ref{theorem1}}\label{appendix:theorem1}

Let $\Delta^{t}$ denote the network energy efficiency in time interval $t \in \mathcal{T}$ and define 
    \begin{equation}
    	\varsigma_n^t = \sum_{m \in \mathcal M}{\sum_{i \in \mathcal I_{m,n}}{ \tau p_{i,n}^{t} } }, \,\,\,\, n \in \mathcal N,
    \end{equation}
and 
    \begin{equation}
    	\chi_{n}^{t} = \frac{ \sum_{ m \in \mathcal M}{ \sum_{ i \in \mathcal{I}_{m,n}}{w_{i,n}^{t}} }  }{\sum_{ m \in \mathcal M}{\sum_{ i \in \mathcal{I}_{m,n}}{w_{i,n}^{t} \eta_{n} }} }, \,\,\,\, n \in \mathcal N. 
    \end{equation}
The network energy efficiency in time interval $t$ is given by:
	\begin{equation}
    	\Delta^{t} = \sum_{ n \in \mathcal N}{ \frac{ \lambda_{n} \chi_{n}^{t} }{ \varsigma_n^t  } }.
    \end{equation}
The Hessian matrix of the network energy efficiency $\Delta^t$ can be written as the following block matrix:
	\begin{equation}
        \begin{aligned} 
    		& \nabla^2 \Delta^t = \\
            & \frac{\partial^2 \Delta^t}{\partial p_{x,y}^{t} \partial p_{x',y'}^{t}}   = \left[\begin{matrix}
    			 \mathbf{A}_{1}^t & & & & \mathbf{0}\\
    			 &	\ddots&	& &\\
    			 &	& \mathbf{A}_{n}^t &	&\\
    			 &	&	& \ddots&	\\
    			 \mathbf{0} & &	& & \mathbf{A}_{N}^t\\
    		\end{matrix}\right]_{IN \times IN}, 
        \end{aligned} 
	\end{equation}
where block $\mathbf{A}_{n}^t$ for any $n \in \mathcal N$ is given by:
    \begin{equation}
       \mathbf{A}_{n}^t = \frac{2 \lambda_{n} \chi_{n}^t  \tau^2 }{ (\varsigma_n^t)^{3} } \begin{bmatrix}
            1 & 1 & \cdots & 1 \\
            1 & 1 & \cdots & 1 \\
            \vdots & \vdots & \ddots & \vdots\\
            1 & 1 & \cdots& 1 
        \end{bmatrix}_{I \times I} 
    \end{equation}

If constraint (\ref{p1}d) is satisfied, $\varsigma_n^t$ is positive. In this case, the first-order leading principal minor of the Hessian matrix, i.e., $\frac{2 \lambda_{n} \chi_{n}^{t} \tau^2}{ (\varsigma_n^t)^{3}}$, is nonnegative. Meanwhile, all the other leading principal minors equal $0$. As a result, the Hessian matrix is positive semidefinite when constraint (\ref{p1}d) is satisfied. Thus, when $\forall p_{i,n}^{t} > 0$, the function $\Delta^t$ is convex. 

Function~$\Delta^t$ increases with the decrease of allocated transmission power for all grids, while the allocated transmission power for all grids should satisfy the SINR constraints in~\eqref{eq9}. Consequently, due to the convexity of function~$\Delta^t$, the IM solution must exist on the boundary of the feasible domain. Thus, the optimal IM solution should satisfy~\eqref{eq9} with equality, i.e.,

	\begin{equation}\label{eq22p}
	   \frac{\bar{p}_{i,n}^{t} h_{i,n,i,n}^{t} }{N_0 + I_{i,n}^{t} } = \rho \gamma_{n}^\mathrm{min}.
    \end{equation}
Define $\hat{\mathbf{H}}^t$ and $\mathbf{\Omega}_{n,n'}^{t}$ in~\eqref{eq18p} and~\eqref{eq13}, respectively. We rewrite~\eqref{eq22p} into the matrix format as~\eqref{eq30}. Therefore, the optimal downlink transmission power in time interval~$t$ can be derived in closed-form  as~\eqref{eq12}.

    \begin{figure*}[t] 
        \begin{equation}\label{eq30}
            \begin{split}
            	& \hat{\mathbf{H}}^t  \cdot \mathbf{p}^t = \rho \left( \left[ \begin{matrix}
            		\gamma _{1}^{\min}\mathbf{\Omega }_{1,1}^{t} &		\cdots&		\gamma _{1}^{\min}\mathbf{\Omega }_{1,n'}^{t}&		\cdots&		\gamma _{1}^{\min}\mathbf{\Omega }_{1,N}^{t}\\
            		\vdots&		\ddots&		\vdots&		&		\vdots\\
            		\gamma _{n}^{\min}\mathbf{\Omega }_{n,1}^{t} &		\cdots&		\gamma _{n}^{\min}\mathbf{\Omega }_{n,n'}^{t} &		\cdots&		\gamma _{n}^{\min}\mathbf{\Omega }_{n,N}^{t} \\
            		\vdots&		&		\vdots&		\ddots&		\vdots\\
            		\gamma _{N}^{\min}\mathbf{\Omega }_{N,1}^{t} &		\cdots&		\gamma _{N}^{\min}\mathbf{\Omega }_{N,n'}^{t} &		\cdots&		\gamma _{N}^{\min}\mathbf{\Omega}_{N,N}^{t}\\
            		\end{matrix} \right]_{IN \times IN} \mathbf{p}^t + \left[ \begin{matrix}
            		\gamma_{1}^{\min} N_0\\
            		\vdots\\
            		\gamma_{n}^{\min} N_0\\
            		\vdots\\
            		\gamma_{N}^{\min} N_0
            		\end{matrix} \right]_{IN \times 1} \right).
            \end{split}
        \end{equation}
        \rule[-10pt]{18.15cm}{0.05em}  
    \end{figure*}

\bibliography{ref}

\bibliographystyle{IEEEtran}

\end{document}